\documentclass[12pt]{article}

\usepackage[T1]{fontenc}
\usepackage{lmodern}
\usepackage{microtype}
\usepackage[a4paper,margin=27mm]{geometry}
\usepackage{amsmath,amssymb,amsthm,mathtools,bm,mathrsfs}
\usepackage{booktabs,tabularx,array,longtable}
\usepackage{enumitem}
\usepackage[numbers,sort&compress]{natbib}
\usepackage[hidelinks]{hyperref}
\newcolumntype{Y}{>{\raggedright\arraybackslash}X}

\newcommand{\eps}{\varepsilon}
\newcommand{\ii}{\mathrm i}
\newcommand{\dd}{\,\mathrm d}
\newcommand{\Hilb}{\mathcal H}
\newcommand{\Kop}{\mathcal K}
\newcommand{\cA}{\mathcal A}
\newcommand{\Span}{\operatorname{span}}
\newcommand{\ad}{\operatorname{ad}}
\newcommand{\abs}[1]{\lvert #1\rvert}

\theoremstyle{plain}
\newtheorem{theorem}{Theorem}[section]

\newtheorem{lemma}[theorem]{Lemma}
\newtheorem{corollary}[theorem]{Corollary}
\newtheorem{criterion}[theorem]{Zakharov--Schulman necessary criterion}
\theoremstyle{remark}
\newtheorem{remark}[theorem]{Remark}

\numberwithin{equation}{section}
\setlist{leftmargin=2em,itemsep=0.25em,topsep=0.35em}
\allowdisplaybreaks
\newif\ifanonymous
\anonymousfalse
\ifanonymous
\hypersetup{
  pdftitle={Six-wave scattering in Hamiltonian Dysthe equations: Euler matching and integrability obstructions},
  pdfsubject={Six-wave scattering, Euler-Dysthe matching, and integrability obstructions},
  pdfkeywords={Hamiltonian Dysthe equations, deep-water Euler equations, six-wave resonance, canonical normal form, integrability obstruction, Zakharov--Schulman criterion},
  pdfauthor={}
}
\else
\hypersetup{
  pdftitle={Six-wave scattering in Hamiltonian Dysthe equations: Euler matching and integrability obstructions},
  pdfsubject={Six-wave scattering, Euler-Dysthe matching, and integrability obstructions},
  pdfkeywords={Hamiltonian Dysthe equations, deep-water Euler equations, six-wave resonance, canonical normal form, integrability obstruction, Zakharov--Schulman criterion},
  pdfauthor={Alex J. Sutherland and Solomon C. Yim}
}
\fi

\title{Six-wave scattering in Hamiltonian Dysthe equations:
Euler matching and integrability obstructions}
\ifanonymous
\author{}
\else
\author{Alex J. Sutherland\thanks{Corresponding author:
\texttt{suthalex@oregonstate.edu}}\quad and\quad Solomon C. Yim\\[0.35em]
\small Department of Civil and Construction Engineering,
Oregon State University\\
\small Corvallis, Oregon 97331, United States\\
\small\href{https://orcid.org/0000-0002-1364-5663}
{Alex J. Sutherland ORCID: 0000-0002-1364-5663}}
\fi
\date{}

\begin{document}
\maketitle

\begin{abstract}
We compute the effective six-wave interaction in spatial and temporal
Hamiltonian Dysthe equations and compare it with the complete degree-six
coefficient of the deep-water Euler Hamiltonian.  On every compact subset
\(K\) of a generic quadratic-dispersion resonance patch, exact resonances
persist under a small cubic correction to the dispersion, and
\[
 W_6\bigl(\Phi_{\lambda,\eps}(z)\bigr)
 =\eps\frac{4BD}{aL(z)}+O_K(\eps^2)
\]
uniformly in the resonance point \(z\) and the deformation parameter
\(\lambda\).  Here \(L\) is the spectral diameter and \(B,D\) are the local
cubic and mean-flow coefficients.  Thus \(BD\ne0\) gives nonzero scattering
on relatively open resonance sets.  At the Craig--Guyenne--Sulem temporal
coefficients, the corresponding exact Euler resonances satisfy
\[
 \frac{\mu}{4}W_6^{\rm E}\bigl(1+\mu\Psi_\mu(z)\bigr)
 =-\frac{8}{L(z)}+O_K(\mu),
\]
where \(\mu\) is the relative bandwidth.  The leading term comes from nine
contractions of quartic vertices; every other connected degree-six
contribution is uniformly bounded.  We classify the coefficient values for
which the generic six-wave coefficient vanishes identically; in the temporal
case, this also gives cancellation on every non-permutation four-wave
resonance.
Outside these sets, the Zakharov--Schulman condition restricts every regular
\(C^3\) quadratic leading symbol to
\(\operatorname{span}\{1,k,\omega(k)\}\), the span of the symbols of wave
action, momentum, and linear energy.  Consequently the physical Dysthe
models cannot support an inverse-scattering hierarchy with infinitely many
linearly independent regular \(C^3\) quadratic leading symbols.
The cancellation sets contain known nonlinear Schr\"odinger, mixed
Chen--Lee--Liu, Calogero--Moser derivative nonlinear Schr\"odinger, and
Hirota representatives.
\end{abstract}



\section{Introduction and main results}

The derivation of deep-water envelope equations is itself rooted in
Hamiltonian normal-form reduction.  Canonical transformations eliminate
nonresonant interactions, while the surviving resonant terms determine the
effective modulation dynamics \cite{Zakharov1968,Krasitskii1994,
CraigWorfolk1995}.  In the narrowband, weakly nonlinear, unidirectional
regime, the leading modulation equation is the cubic nonlinear Schr\"odinger
(NLS) equation.  Its exceptional scattering cancellations accompany an
inverse-scattering structure that organizes solitons, commuting conservation
laws, and nonlinear spectral evolution
\cite{ZakharovShabat1972,ZakharovManakov1974}.

The Dysthe approximation retains higher-order corrections omitted by NLS,
most notably derivative nonlinearities and the modulation-induced
mean flow and, in temporal formulations, a cubic correction to linear
dispersion \cite{Dysthe1979}.  Spatial and temporal Dysthe equations are
different asymptotic evolution descriptions of the same underlying
narrowband deep-water wave problem: a spatial equation evolves a measured
time series downstream, whereas a temporal equation evolves a spatial profile
in time \cite{LoMei1985,TrulsenDysthe1996,KitShemer2002}.  Their canonical
variables and complete coefficient tuples need not be identical.  Hamiltonian
formulations nevertheless put their resonant interactions in a common
normal-form language \cite{CraigGuyenneSulem2010,GramstadTrulsen2011,
CraigGuyenneSulem2012,GuyenneEtAl2021Spatial}.

We ask whether these corrections preserve the resonant cancellations of
integrable NLS.  For the Hamiltonian Dysthe models considered here, they do
not.  We first show that the six-wave coefficient remains nonzero on open
resonance patches as cubic dispersion is introduced.  We then show that, in
the narrowband limit, the complete degree-six Euler coefficient has the same
leading term.  Finally, we classify the coefficient values for which the
generic six-wave coefficient vanishes and determine the associated four-wave
cancellation condition.  The small-dispersion limit is
singular globally because an additional four-wave resonance plane moves to
unbounded wavenumber.

For quadratic dispersion, the first non-permutation resonant interactions in
the gauge-invariant family considered here occur at connected six-wave order.
Cubic dispersion also permits a non-permutation four-wave resonance component,
which is treated separately in Section~\ref{subsec:four-wave-plane}.  After
nonresonant quartic terms are removed by a canonical Lie transform
\cite{Deprit1969,Krasitskii1994}, the interaction between three incoming and
three outgoing Fourier modes is represented by a coefficient \(W_6\).  On an
exact \(3\to3\) resonance the linear frequency mismatch vanishes, so this
coefficient cannot be removed by the same regular homological equation.
Zakharov and Schulman showed that the quadratic leading symbol of an
additional conserved quantity must be additive across such a collision unless
the connected scattering coefficient vanishes
\cite{ZakharovSchulman1980,ZakharovSchulman1988}.  This gives a necessary
condition for an inverse-scattering hierarchy whose conserved quantities
satisfy the regularity assumptions in Section~\ref{sec:ZS}.

\subsection{Uniform persistence on resonance patches}

The central theorem establishes uniform noncancellation on relatively open
resonance patches.  In the Dysthe scaling
\[
 \omega_{\lambda,\eps}(k)=ak^2+\lambda\eps hk^3,
 \qquad c=\eps C,\qquad d=\eps D,\qquad e=\eps^2E,
 \qquad 0\leq\lambda\leq1,
\]
every compactly contained generic open patch of the quadratic-dispersion
\(3\to3\) resonance manifold has an analytic continuation consisting of exact
cubic-dispersive resonances.  The continuation preserves the
absolute-value chamber, keeps all nine virtual denominators uniformly away
from zero, and satisfies
\begin{equation}
 W_6\bigl(\Phi_{\lambda,\eps}(z)\bigr)
 =\eps\frac{4BD}{aL(z)}+O_K(\eps^2)
 \label{eq:intro-open-patch}
\end{equation}
uniformly in both the resonance point \(z\) and the deformation parameter
\(\lambda\).  If the external modes at \(z\) are
\(p_1,p_2,p_3,q_1,q_2,q_3\), their spectral diameter is
\begin{equation}
 L(z)=\max\{p_1,p_2,p_3,q_1,q_2,q_3\}
      -\min\{p_1,p_2,p_3,q_1,q_2,q_3\}>0.
 \label{eq:intro-spectral-diameter}
\end{equation}
Thus, when \(BD\ne0\), noncancellation persists for every sufficiently small
nonzero \(\eps\), uniformly in \(\lambda\in[0,1]\), on an entire relatively
open resonance patch, not merely along one explicit test sextet.  The precise
statement is Theorem~\ref{thm:strong-open}; its first-variation coefficient is
intrinsic, although the proof uses an explicit polynomial identification of
the two resonance manifolds.

Formula \eqref{eq:intro-open-patch} isolates the leading mechanism: the coupling
between the local cubic nonlinearity and the nonlocal mean-flow term,
represented by the corresponding contributions \(B\) and \(D\) to the quartic
interaction kernel.  Cubic dispersion changes the resonance geometry but does
not change this first variation.

For the Euler comparison we use the unit-carrier normalization of Craig,
Guyenne, and Sulem
\cite{CraigGuyenneSulem2021}.  The temporal coefficients are
\begin{equation}
 (a,h,B,C,D,E)=\left(-\frac18,\frac1{16},1,3,1,0\right).
 \label{eq:intro-CGS-source}
\end{equation}
Let \(\mu>0\) denote relative bandwidth, distinct from the general Dysthe
bookkeeping parameter \(\eps\).  Theorem~\ref{thm:euler-dysthe-matching}
constructs analytic exact-resonance embeddings \(\Psi_\mu\) for Euler and
\(\Phi_\mu\) for this temporal Dysthe model and proves
\begin{align}
 W_6^{\rm E}\bigl(1+\mu\Psi_\mu(z)\bigr)
 &=\mu^{-2}\widehat W_6^{\rm D}\bigl(\Phi_\mu(z)\bigr)+O_K(1),
 \label{eq:intro-strong-euler-match}\\
 \frac{\mu}{4}W_6^{\rm E}\bigl(1+\mu\Psi_\mu(z)\bigr)
 &=-\frac8{L(z)}+O_K(\mu).
 \label{eq:intro-normalized-euler-match}
\end{align}
The first equality compares complete connected coefficients, while the
second evaluates the symbolic Dysthe law \(BD/(aL)\) at
\eqref{eq:intro-CGS-source}.  The matching is uniform on every compact subset
of a generic fixed chamber.  It shows that the Dysthe obstruction is the
leading narrowband contribution of the parent Euler Hamiltonian, rather than
an interaction produced only by the envelope truncation.

\subsection{Classification and physical consequences}

For quadratic dispersion we compute the complete generic kernel
\[
 W_{6,\mathrm{sp}}
 =\frac{4d(B+c\chi)}{aL}-\frac{3d^2}{a}-12e
\]
and prove that it vanishes identically only on two coefficient components.
For cubic-corrected dispersion we identify the additional four-wave resonance
plane and classify the coefficient values for which the six-wave coefficient
vanishes identically.  These values also give four-wave cancellation on the
additional resonance plane.  The resulting cancellation sets contain, after
normalization, representatives of
three independently known integrable families: the mixed Chen--Lee--Liu
derivative nonlinear Schr\"odinger equation, the Calogero--Moser derivative
nonlinear Schr\"odinger equation, and the Hirota equation, which is a linear
combination of the NLS and complex modified Korteweg--de Vries flows
\cite{ChenLeeLiu1979,Hirota1973,
GerardLenzmann2024,Badreddine2024}.  Their integrability comes from the cited
literature; the present calculation establishes the stated resonant
cancellations.

The spatial coefficient tuples of Fedele and Dutykh
\cite{FedeleDutykh2011,FedeleDutykh2012Extended} and the temporal Hamiltonian
Dysthe equation of Craig, Guyenne, and Sulem, with its linear dispersion
truncated at cubic order \cite{CraigGuyenneSulem2021}, lie outside their
corresponding cancellation loci.
Corollaries~\ref{cor:spatial-physical-obstruction} and
\ref{cor:temporal-physical-obstruction} show that every regular \(C^3\)
quadratic leading symbol is a linear combination of
\(1,k,\omega(k)\).  Thus these models cannot support an
inverse-scattering hierarchy with infinitely many linearly independent
regular \(C^3\) quadratic leading symbols.  Section~\ref{sec:ZS}
states the precise regularity assumptions.  The Euler theorem is a separate
comparison of finite normal-form coefficients.

Throughout, subscripts on \(H_j\) count homogeneous Fourier-amplitude degree,
whereas \(\eps\) denotes steepness and is held fixed during the normal-form
calculation.

\subsection{Relation to prior six-wave calculations and organization}

Fedele and Dutykh reported inelastic solitary-wave collisions as numerical
evidence suggestive of nonintegrability
\cite{FedeleDutykh2011,FedeleDutykh2012Extended}, but did not calculate the
connected Dysthe six-wave coefficient.  Craig and Worfolk found an integrable
fourth-order normal form for infinite-depth water waves and a nonzero higher
resonant obstruction \cite{CraigWorfolk1995}.  Related scattering calculations
include the \(3\to3\), nine-channel amplitude for a compact unidirectional
Zakharov equation \cite{DyachenkoKachulinZakharov2013}.  That calculation
retained the pair-of-quartic contractions and explicitly left open whether the
direct six-wave term of the exact water-wave Hamiltonian could change the total
coefficient.  Theorem~\ref{thm:euler-dysthe-matching} addresses that issue in
the narrowband regime: the complete connected coefficient is included and
every sector omitted by the quartic truncation is proved uniformly bounded, so
none can cancel the singular leading term.  Other related calculations include
the recent closed all-multiplicity one-dimensional gravity-wave amplitudes
organized by chamber geometry \cite{ArkaniHamedEtAl2026}.  Subsequent work
traced the hydrotope geometry directly to the water-wave action
\cite{CaoEtAl2026}.
Their complete tree amplitudes apply to the two-minus sector.  In the
all-incoming convention of \cite{ArkaniHamedEtAl2026}, \emph{minus} denotes
negative spatial momentum, not negative frequency: the present positive-carrier
\(3\to3\) coefficient has momenta
\((p_1,p_2,p_3,-q_1,-q_2,-q_3)\) and is a three-minus six-point amplitude,
a sector left open there.  Cao et al.\ derive denominator-free contact
coefficients for general momentum-sign assignments, but their complete tree
resummation is likewise for the two-minus sector.  The overlap with our
calculation is therefore structural---exact Euler tree dressing, chamber
geometry, and cancellation of apparent poles associated with internal
denominators---rather than an identification of amplitudes.  Related work on
resonant water-wave Hamiltonians
and integrability tests includes \cite{Janssen2009,Ussembayev2019,
LinHuang2019,AblowitzLuoMusslimani2023}.

The exact water-wave Hamiltonian, Dirichlet--Neumann expansion, cubic normal
form, and temporal Dysthe quartic coefficient are prior ingredients
\cite{Zakharov1968,CraigSulem1993,CraigSulem2016,
CraigGuyenneSulem2021,Krasitskii1994}.  The Euler-side novelty here is the
uniform compact-patch matching of the \emph{complete connected} six-wave
coefficient, including the proof that the direct sextic vertex, all connected
contributions involving odd-degree Hamiltonian vertices, and every nonmatching
quartic sector remain uniformly
bounded in the narrowband limit.

Sections~\ref{sec:family}--\ref{sec:ZS} fix the normalization and obstruction
criterion; Sections~\ref{sec:spatial}--\ref{sec:bifurcation} prove the two
classifications, uniform persistence on resonance patches, Euler matching, and
the singular limit of the cancellation loci.  Exact certificates and
reproducibility information follow the
analytic arguments.

\subsection{Notation guide}

Table~\ref{tab:notation} collects the symbols that change role between the
spatial, temporal, and Euler parts of the paper.
\begin{table}[htbp]
\centering
\small
\begin{tabularx}{\textwidth}{@{}p{0.24\textwidth}Y@{}}
\toprule
Symbol & Meaning\\
\midrule
\(k\) & Fourier variable dual to the profile coordinate: a wavenumber in
temporal evolution and a frequency detuning in spatial evolution\\
\(p=(p_1,p_2,p_3)\), \(q=(q_1,q_2,q_3)\) & Incoming and outgoing
\(3\to3\) spectral triples\\
\(T\), \(W_6\) & Ordered quartic interaction kernel and effective connected
six-wave coefficient\\
\(a,\gamma\) & Quadratic and cubic coefficients of
\(\omega_\gamma(k)=ak^2+\gamma k^3\)\\
\(B,c,d,e\) & Local cubic, derivative-local, mean-flow, and local quintic
Hamiltonian coefficients\\
\(\eps,\lambda\) & Dysthe scaling parameter and dispersion-deformation
parameter, with \(\gamma=\lambda\eps h\)\\
\(\mu\) & Relative bandwidth about the unit Euler carrier; distinct from
\(\eps\) until the temporal comparison\\
\(K\Subset U\) & Compact subset of a connected generic fixed-chamber
resonance patch\\
\(L\) & Spectral diameter of the six external modes; the leading laws are
\(4BD/(aL)\) and \(-8/L\)\\
\(\Phi_{\lambda,\eps}\), \(\Phi_\mu\), \(\Psi_\mu\) & Exact-resonance
identifications for the Dysthe deformation, temporal Dysthe model, and Euler
dispersion\\
\(A_\kappa,U_\kappa\) & \(a+3\gamma\kappa\) and \(B+c\kappa\) at a collision
center \(\kappa\)\\
\bottomrule
\end{tabularx}
\caption{Principal notation and its role.}
\label{tab:notation}
\end{table}

\section{Canonical Hamiltonian family and Fourier representation}
\label{sec:family}

\subsection{Fourier convention and bracket}

Let \(u(t,x)\) be the complex envelope, with Fourier representation
\begin{equation}
 u(x)=\frac{1}{\sqrt{2\pi}}
 \int_{\mathbb R}\alpha(k)e^{\ii kx}\dd k.
 \label{eq:Fourier}
\end{equation}
We use the temporal canonical convention
\begin{equation}
 \ii u_t=\frac{\delta H}{\delta\bar u},
 \label{eq:canonical}
\end{equation}
or equivalently \(u_t=\{u,H\}\), with
\begin{equation}
 \{F,G\}
 =-\ii\int_{\mathbb R}
 \left(
 \frac{\delta F}{\delta\alpha(k)}
 \frac{\delta G}{\delta\bar\alpha(k)}
 -
 \frac{\delta F}{\delta\bar\alpha(k)}
 \frac{\delta G}{\delta\alpha(k)}
 \right)\dd k.
 \label{eq:bracket}
\end{equation}
Reversing the evolution variable puts the spatial convention
\(v_\xi=\ii\delta H/\delta\bar v\) into the same canonical form:
\begin{equation}
 t=-\xi,\qquad x=\tau,\qquad u(t,x)=v(-t,x).
 \label{eq:reversal}
\end{equation}
Throughout, \(k\) denotes the Fourier variable dual to the profile coordinate.
It is a wavenumber in the temporal formulation and a frequency detuning in the
original spatial formulation.  Accordingly, \(P\) below generates translations
in the chosen profile coordinate.

\subsection{Physical-space Hamiltonian and interaction kernels}

Let
\begin{equation}
 \Kop=\Hilb\partial_x,
 \qquad
 \widehat{\Kop f}(k)=\abs{k}\widehat f(k).
 \label{eq:K}
\end{equation}
Here \(\Hilb\) is the Hilbert transform with Fourier symbol
\(-\ii\operatorname{sgn}k\).  Thus
\(\Kop=|D_x|\), where \(D_x=-\ii\partial_x\).
Consider
\begin{equation}
\begin{aligned}
H[u]=\int_{\mathbb R}\bigg[
&a\abs{u_x}^2
-\frac{\ii\gamma}{2}(\bar u_xu_{xx}-u_x\bar u_{xx})
+\frac B2\abs u^4\\
&+\frac{\ii c}{4}\abs u^2(\bar u_xu-u_x\bar u)
-\frac d2\abs u^2\Kop(\abs u^2)
-\frac e3\abs u^6
\bigg]\dd x,
\end{aligned}
\label{eq:Hamiltonian}
\end{equation}
where \(a,B,c,d,e,\gamma\in\mathbb R\) and \(a\ne0\).  Every displayed density
is real: the derivative brackets are purely imaginary before multiplication
by \(\ii\), while the real even multiplier \(\abs{k}\) makes \(\Kop\)
self-adjoint.  Fourier and variational calculations are performed formally
in the amplitudes; integration by parts may be evaluated on smooth, rapidly
decaying test functions.

The corresponding PDE is
\begin{equation}
\begin{aligned}
\ii u_t
{}&+a u_{xx}-\ii\gamma u_{xxx}-B\abs u^2u
+\ii c\abs u^2u_x\\
&+d\,u\Kop(\abs u^2)+e\abs u^4u=0.
\end{aligned}
\label{eq:PDE}
\end{equation}

\subsection{Main resonance-patch theorem}

For the dispersion
\[
 \omega_{\lambda,\eps}(k)=ak^2+\lambda\eps hk^3
\]
let \(\mathcal R_{\lambda,\eps}\) denote the labeled exact \(3\to3\)
resonance manifold.  Let \(\mathcal R_0^{\rm gen}\) be its
quadratic-dispersion member after removing permutation resonances, external
collisions, rank-degenerate points, absolute-value chamber walls, repeated
virtual slots, and zero virtual denominators.  For \(z=(p,q)\), write
\[
 L(z)=\max(p_1,p_2,p_3,q_1,q_2,q_3)
      -\min(p_1,p_2,p_3,q_1,q_2,q_3).
\]

\begin{theorem}[Uniform persistence under resonance deformation]
\label{thm:strong-open}
Fix \(a\ne0\) and real \(h,B,C,D,E\), set
\[
 c=\eps C,\qquad d=\eps D,\qquad e=\eps^2E,
\]
and let \(0\leq\lambda\leq1\).  Let \(U\) be a connected fixed-chamber
relatively open subset of \(\mathcal R_0^{\rm gen}\), and let
\(K\Subset U\).

\begin{enumerate}[label=\textup{(\roman*)}]
\item There is \(\eps_K>0\) and a real-analytic embedding
\[
 \Phi_{\lambda,\eps}:K\longrightarrow\mathcal R_{\lambda,\eps},
 \qquad |\eps|<\eps_K,
\]
which is the restriction of an analytic local diffeomorphism on a relative
neighborhood of \(K\).  It satisfies
\(\Phi_{\lambda,0}=\operatorname{id}\); all image resonances remain in the
same chamber; and all nine internal denominators are uniformly bounded away
from zero.  Moreover,
\begin{equation}
 W_6\bigl(\Phi_{\lambda,\eps}(z)\bigr)
 =\eps\frac{4BD}{aL(z)}+O_K(\eps^2)
 \label{eq:open-patch-expansion}
\end{equation}
uniformly for \((z,\lambda)\in K\times[0,1]\).

\item If \(V\Subset U\) is relatively open and \(K=\overline V\), then
\(\Phi_{\lambda,\eps}(V)\) is relatively open in
\(\mathcal R_{\lambda,\eps}\).  If \(BD\ne0\), there is \(c_K>0\) such that,
for every sufficiently small nonzero \(\eps\),
\[
 \left|W_6\bigl(\Phi_{\lambda,\eps}(z)\bigr)\right|
 \geq c_K|\eps|,\qquad z\in V,\quad 0\leq\lambda\leq1.
\]
\end{enumerate}
\end{theorem}

\begin{remark}
Theorem~\ref{thm:strong-open} is a finite normal-form coefficient theorem;
it uses no Zakharov--Schulman hierarchy hypothesis.  The condition \(BD\ne0\)
is needed only for nonvanishing.  If \(h=0\), the analytic statement remains
true but the endpoint does not introduce cubic dispersion.
\end{remark}

\subsection{First variations and derivation of the PDE}

Treating \(u\) and \(\bar u\) as independent variables, integration by parts
and self-adjointness of \(\Kop\) give
\begin{equation}
\frac{\delta H}{\delta\bar u}
=-au_{xx}+\ii\gamma u_{xxx}+B|u|^2u
-\ii c|u|^2u_x-du\Kop(|u|^2)-e|u|^4u,
\label{eq:full-variation}
\end{equation}
which is equivalent to \eqref{eq:PDE}.  A term-by-term calculation is included
in Online Resource~1.

In Fourier coordinates,
\begin{equation}
 H_2=\int_{\mathbb R}\omega_\gamma(k)\abs{\alpha(k)}^2\dd k,
 \qquad
 \omega_\gamma(k)=ak^2+\gamma k^3.
 \label{eq:H2}
\end{equation}
We call \(\omega_\gamma(k)=ak^2+\gamma k^3\) the cubic-corrected dispersion.
Define
\begin{equation}
 d\Gamma_4
 =\frac{1}{2\pi}
 \delta(k_1+k_2-k_3-k_4)\,dk_1dk_2dk_3dk_4.
 \label{eq:Gamma4}
\end{equation}
Then
\begin{equation}
 H_4=\frac12\int
 T(k_1,k_2;k_3,k_4)
 \bar\alpha_1\bar\alpha_2\alpha_3\alpha_4\dd\Gamma_4,
 \label{eq:H4}
\end{equation}
where
\begin{equation}
\boxed{
 T(k_1,k_2;k_3,k_4)
 =B+\frac c2(k_1+k_2)
 -\frac d4
 \sum_{\substack{r=1,2\\s=3,4}}\abs{k_r-k_s}.
}
\label{eq:vertex}
\end{equation}
\subsection{Derivation of the quartic Fourier vertex}

Four Fourier factors contribute \((2\pi)^{-2}\), while spatial integration
contributes \(2\pi\delta(k_1+k_2-k_3-k_4)\). Thus every quartic term carries
the measure \(d\Gamma_4\). The local term gives the constant vertex \(B\).

For the derivative term,
\(\partial_x\bar u(k)=-\ii k\bar u(k)\) and
\(\partial_xu(k)=\ii ku(k)\). Symmetrizing over the two barred and two
unbarred legs gives
\[
\frac c4(k_1+k_2+k_3+k_4)
=\frac c2(k_1+k_2)
\]
on the momentum delta function. Finally, \(\rho=\bar uu\) has Fourier
transfer \(k_{\rm unbar}-k_{\rm bar}\). Expansion of
\(-d\int\rho\Kop\rho/2\), using the multiplier \(|k|\), followed by the same
within-pair symmetrization gives
\[
-\frac d4\bigl(
|k_1-k_3|+|k_1-k_4|+|k_2-k_3|+|k_2-k_4|
\bigr).
\]
Adding these contributions proves \eqref{eq:vertex}, including the
\(1/(2\pi)\) measure and separate barred/unbarred symmetry.
Translation invariance supplies the momentum delta function, while the
quadratic Hamiltonian supplies the frequency mismatch
\begin{equation}
 \Delta_4
 =\omega_\gamma(k_1)+\omega_\gamma(k_2)
 -\omega_\gamma(k_3)-\omega_\gamma(k_4).
 \label{eq:mismatch}
\end{equation}
We call \(T\) the quartic interaction kernel, or quartic vertex.  After the
normal-form transformation, \(W_6\) is the effective six-wave coefficient;
on exact resonances it is the formal six-wave scattering amplitude in this
normalization.

For sextic terms set
\begin{equation}
 d\Gamma_6
 =\frac{1}{(2\pi)^2}
 \delta(p_1+p_2+p_3-q_1-q_2-q_3)
 \prod_{i=1}^3dp_i\prod_{j=1}^3dq_j,
 \label{eq:Gamma6}
\end{equation}
We integrate over ordered sextuples and take \(W_6\) to be symmetric separately
in the incoming and outgoing variables, with the prefactor \(1/(3!3!)\):
\begin{equation}
 H_6=\frac{1}{3!3!}\int
 W_6(p;q)
 \prod_{j=1}^3\bar\alpha(q_j)
 \prod_{i=1}^3\alpha(p_i)\dd\Gamma_6.
 \label{eq:H6}
\end{equation}
Consequently the bare term \(-e\int\abs u^6/3\) contributes
\begin{equation}
 W_6^{\mathrm{bare}}=-12e.
 \label{eq:bare}
\end{equation}

\subsection{Symmetries and quadratic parts of the basic conserved quantities}

All terms contain equal numbers of \(u\) and \(\bar u\), so \(H\) is
invariant under the global \(U(1)\) action \(u\mapsto e^{\ii\theta}u\).
No density depends explicitly on \(x\), so \(H\) is translation invariant.
The associated quadratic generators are
\begin{equation}
 N=\int\abs{\alpha(k)}^2\dd k,
 \qquad
 P=\int k\abs{\alpha(k)}^2\dd k.
 \label{eq:NP}
\end{equation}
The quadratic parts of \(N\), \(P\), and \(H\) have symbols
\(1\), \(k\), and \(\omega_\gamma(k)\).

\subsection{Action-only monomials}

For a discrete Fourier approximation define
\[
 J_k=\bar\alpha_k\alpha_k=\abs{\alpha_k}^2.
\]
A gauge-invariant \(m\to m\) monomial is
\[
 M_{p,q}=\bar\alpha_{q_1}\cdots\bar\alpha_{q_m}
 \alpha_{p_1}\cdots\alpha_{p_m}.
\]
If the incoming and outgoing multisets agree, including multiplicities, a
permutation \(\sigma\) satisfies \(q_j=p_{\sigma(j)}\).  Commutativity gives
\begin{equation}
 M_{p,q}=\prod_{j=1}^m\bar\alpha_{p_j}\alpha_{p_j}
 =\prod_{j=1}^mJ_{p_j}.
 \label{eq:action-only}
\end{equation}
It is therefore a function only of the actions.  The multiset qualification
handles repeated modes correctly.

\section{Normal form through degree six}
\label{sec:normal-form}

Write
\begin{equation}
 H=H_2+H_4+H_6^{\mathrm{bare}}+O(8),
 \qquad H_4=Z_4+P_4,
 \label{eq:split}
\end{equation}
where \(Z_4\) is resonant and \(P_4\) is nonresonant.  For a quartic monomial
\(M_Q\), direct use of the bracket gives
\[
 \{H_2,M_Q\}=-\ii\Delta_QM_Q.
\]
Consequently, if \(C_QM_Q\) is a term of \(P_4\), then
\[
 (F_4)_Q=-\ii\,\frac{C_Q}{\Delta_Q}M_Q
\]
satisfies \(\{H_2,(F_4)_Q\}=-C_QM_Q\).  Summing these terms over all
nonresonant quartets gives \(F_4\) with
\begin{equation}
 \{H_2,F_4\}=-P_4.
 \label{eq:homological}
\end{equation}
No division is made on the resonant set.  Apply
\begin{equation}
 \widetilde H=e^{\ad_{F_4}}H,
 \qquad
 \ad_{F_4}G=\{G,F_4\}.
 \label{eq:Lie}
\end{equation}
Since \(\deg\{G_m,F_n\}=m+n-2\), the expansion through degree six is
\begin{equation}
\boxed{
 \widetilde H
 =H_2+Z_4+H_6^{\mathrm{bare}}
 +\{Z_4,F_4\}
 +\frac12\{P_4,F_4\}+O(8).
}
\label{eq:normalform}
\end{equation}

For a disjoint sextic target---six distinct external modes and no mode common
to its incoming and outgoing triples---the action-only part of \(Z_4\) cannot
create the target.  A quartic action monomial \(J_kJ_l\) retains an intact
action factor after one contraction, so at least one wavenumber would
occur on both sides.  In the
temporal problem there is also a non-action quartic resonant plane.  For a
generic target whose nine virtual quartets all have nonzero mismatch, none of
those quartets lies on that plane.  Thus in both calculations the target
coefficient comes from
\[
 H_6^{\mathrm{bare}}+\frac12\{P_4,F_4\}.
\]

Choose the incoming mode \(p_\ell\) omitted from the first quartic vertex and
a selected outgoing mode \(q_\alpha\), and let \(\{p_i,p_j\}\) be the
remaining pair.  The internal
mode is
\begin{equation}
 r_{\ell\alpha}=p_i+p_j-q_\alpha=S-p_\ell-q_\alpha.
 \label{eq:internal}
\end{equation}
There are \(3\times3=9\) channels.  Put
\begin{equation}
 \Delta_{\ell\alpha}
 =\omega_\gamma(q_\alpha)+\omega_\gamma(r_{\ell\alpha})
 -\omega_\gamma(p_i)-\omega_\gamma(p_j),
 \label{eq:channel-den}
\end{equation}
and, writing \(\{\beta,\nu\}=\{1,2,3\}\setminus\{\alpha\}\), define
the slots explicitly by
\[
T^{(1)}_{\ell\alpha}=T(q_\alpha,r_{\ell\alpha};p_i,p_j),
\qquad
T^{(2)}_{\ell\alpha}=T(q_\beta,q_\nu;r_{\ell\alpha},p_\ell).
\]
The formula applies to a disjoint target with six distinct external modes,
four pairwise distinct slots in each virtual quartet, nonzero
\(\Delta_{\ell\alpha}\). Then
\begin{equation}
 \cA(p;q)=\sum_{\ell=1}^3\sum_{\alpha=1}^3
 \frac{T^{(1)}_{\ell\alpha}T^{(2)}_{\ell\alpha}}
 {\Delta_{\ell\alpha}},
 \qquad
 \boxed{W_6=-4\cA-12e.}
 \label{eq:channel-sum}
\end{equation}
The numerical factors are convention-dependent; this is why the Fourier and
factorial normalizations precede the theorem.

For completeness, the factor \(-4\) and the count of nine channels follow
directly from these conventions.  For fixed \((\ell,\alpha)\), let \(M_1\)
and \(M_2\) be the two quartic monomials joined through
\(r_{\ell\alpha}\).  The factor \(1/2\) in \(H_4\), together with the
\(2!2!\) permutations within the barred and unbarred pairs, gives
coefficients \(2T^{(1)}_{\ell\alpha}\) and
\(2T^{(2)}_{\ell\alpha}\).  Their frequency mismatches are
\(\Delta_{\ell\alpha}\) and \(-\Delta_{\ell\alpha}\), while the two
ordered cross terms in \(\{P_4,F_4\}\) contribute
\(-8T^{(1)}_{\ell\alpha}T^{(2)}_{\ell\alpha}/\Delta_{\ell\alpha}\).
The Lie-series factor \(1/2\) leaves the coefficient \(-4\) in
\eqref{eq:channel-sum}.  Each one-contraction decomposition is indexed
uniquely by the omitted incoming mode \(p_\ell\) and selected outgoing mode
\(q_\alpha\), hence \(3\times3=9\) channels.  Contracting the two quartic
momentum deltas produces the sextic delta with no additional Jacobian in the
unitary Fourier convention, and comparison with the \(1/(3!3!)\) sextic
normalization gives the bare contribution \(-12e\).

\section{Zakharov--Schulman necessary condition}
\label{sec:ZS}

Consider a translation- and phase-invariant formal conserved functional
\begin{equation}
 I=I_2+I_4+I_6+\cdots,
 \qquad
 I_2=\int\varphi(k)\abs{\alpha(k)}^2\dd k.
 \label{eq:I}
\end{equation}
where \(\varphi\in C^3(J)\) on the connected spectral interval under
consideration.  Assume that the kernels required through degree six are
regular on the generic resonances used below, where all internal frequency
mismatches are nonzero.  Conservation is imposed coefficient by coefficient
in the Fourier amplitudes.  We call \(\varphi\) the quadratic leading symbol;
quadratic refers to amplitude degree.

\begin{criterion}
On a generic exact \(3\to3\) resonance, regular solvability of the sixth-order
homological equation requires
\begin{equation}
 \boxed{
 D_\varphi(p;q)W_6(p;q)=0,
 \qquad
 D_\varphi(p;q)=\sum_{i=1}^3\varphi(p_i)
 -\sum_{j=1}^3\varphi(q_j).
 }
 \label{eq:ZS}
\end{equation}
\end{criterion}

This is the \(3\to3\), phase-invariant specialization of
Zakharov and Schulman's scattering alternative
\cite[Theorem~3.1]{ZakharovSchulman1988}; see also
\cite{ZakharovSchulman1980}.

In the present normalization, transform \(H\) and \(I\) by the same quartic
Lie map. At degree four one solves
\[
\{H_2,\widetilde I_4\}+\{Z_4,I_2\}=0
\]
where the internal frequency mismatches are nonzero.  At degree six a
generic disjoint resonant
monomial has no \(\{H_2,\widetilde I_6\}\) contribution because its
\(3\to3\) mismatch is zero. The remaining connected coefficient is
\(D_\varphi(p;q)W_6(p;q)\). Action-only quartic factors cannot create the
disjoint target; on the cubic-dispersive slice, a non-action resonant quartic
contribution would force one of the excluded internal denominators to vanish.
Thus regular solvability gives \(D_\varphi W_6=0\).

\section{Quadratic-dispersion spatial slice}
\label{sec:spatial}

Set \(\gamma=0\), so \(\omega_0(k)=ak^2\).

\subsection{Quartic resonances are action-only}

For quadratic one-dimensional dispersion, four-wave resonances are
permutation resonances; hence the quartic resonant normal form is a function
of the modal actions.  We recall the elementary argument for completeness and
use it only as standard resonance background.

Let \(k_1+k_2=k_3+k_4=S\) and impose quadratic frequency resonance.  Since
\[
 k_1^2+k_2^2=S^2-2k_1k_2,
\]
frequency equality gives \(k_1k_2=k_3k_4\).  The unordered pairs are roots
of the same polynomial \(z^2-Sz+k_1k_2\), hence
\begin{equation}
 \{k_1,k_2\}=\{k_3,k_4\}.
 \label{eq:4action}
\end{equation}
Every resonant quartic monomial is action-only.

\subsection{Exact six-wave geometry}

Let \(p,q\) satisfy
\begin{equation}
 \sum_i p_i=\sum_jq_j=S,
 \qquad
 \sum_i p_i^2=\sum_jq_j^2.
 \label{eq:sp-res}
\end{equation}
Define \(P(z)=\prod_i(z-p_i)\) and \(Q(z)=\prod_j(z-q_j)\).  The triples
have the same first two elementary symmetric sums, so
\begin{equation}
 P(z)-Q(z)=C_0
 \label{eq:constant-difference}
\end{equation}
for \(C_0\ne0\) on a genuine non-permutation resonance.  The level-set
picture of a monic cubic yields only two generic interlacings, related by
exchanging \(p\) and \(q\):
\[
 q_1<p_1<p_2<q_2<q_3<p_3
\]
or its reverse.  The channel denominator is
\begin{equation}
 \Delta_{\ell\alpha}
 =2a(q_\alpha-p_i)(q_\alpha-p_j),
 \qquad
 \frac1{\Delta_{\ell\alpha}}
 =\frac{q_\alpha-p_\ell}{2aC_0}.
 \label{eq:sp-den}
\end{equation}
Set
\begin{equation}
\begin{gathered}
 k_-=\min(p_1,p_2,p_3,q_1,q_2,q_3),\qquad
 k_+=\max(p_1,p_2,p_3,q_1,q_2,q_3),\\
 L=k_+-k_-,\qquad
 \chi=\frac{2S-k_--k_+}{4}.
\end{gathered}
\label{eq:Lchi}
\end{equation}
Thus \(\chi\) is the mean of the four non-extreme external modes.

\subsection{Spatial kernel and cancellation components}

\begin{theorem}[Spatial six-wave kernel]
\label{thm:spatial}
On every generic genuine spatial \(3\to3\) resonance with distinct external
modes and nonzero internal denominators,
\begin{equation}
 \boxed{
 W_{6,\mathrm{sp}}(p;q)
 =\frac{4d(B+c\chi)}{aL}
 -\frac{3d^2}{a}-12e.
 }
 \label{eq:Wsp}
\end{equation}
It vanishes identically on the generic spatial resonance manifold if and only
if either
\begin{equation}
 \boxed{d=e=0}
 \label{eq:local-branch}
\end{equation}
or
\begin{equation}
 \boxed{B=c=0,\qquad e=-\frac{d^2}{4a}.}
 \label{eq:CM-branch}
\end{equation}
\end{theorem}

The first component is the local mixed Chen--Lee--Liu derivative nonlinear
Schr\"odinger component; the
second contains the focusing canonical Calogero--Moser gauge representative
and its defocusing counterpart.  Corollary~\ref{cor:integrable-branches}
gives the transformations, sign qualification, and primary citations.
The algebraic cancellation theorem itself does not prove all-order
integrability.

\begin{proof}
Resolve the absolute values in the first interlacing chamber by writing
\[
\begin{aligned}
p&=(k_0+x,\ k_0+x+y,\ k_0+2x+y+2z+w),\\
q&=(k_0,\ k_0+2x+y+z,\ k_0+2x+y+z+w),
\end{aligned}
\]
where \(x,y,z,w>0\) and quadratic resonance is equivalent to
\(x(x+y)=z(z+w)\).  Here \(L=2x+y+2z+w\) and
\(C_0=-x(x+y)L\).  Define the cross-difference sums
\begin{equation}
 R^{(1)}_{\ell\alpha}=\sum_{m\ne\ell}|q_\alpha-p_m|,
 \qquad
 R^{(2)}_{\ell\alpha}=\sum_{n\ne\alpha}|p_\ell-q_n|,
 \label{eq:spatial-cross-sums}
\end{equation}
and the following finite sums, with \(\ell,\alpha\in\{1,2,3\}\):
\begin{equation}
\begin{aligned}
 M_0&=\sum_{\ell,\alpha}(q_\alpha-p_\ell)
       (R^{(1)}_{\ell\alpha}+R^{(2)}_{\ell\alpha}),\\
 M_1&=\sum_{\ell,\alpha}(q_\alpha-p_\ell)
       [(S-q_\alpha)R^{(1)}_{\ell\alpha}
        +(S-p_\ell)R^{(2)}_{\ell\alpha}],\\
 M_2&=\sum_{\ell,\alpha}(q_\alpha-p_\ell)
       R^{(1)}_{\ell\alpha}R^{(2)}_{\ell\alpha}.
\end{aligned}
\label{eq:spatial-M-definitions}
\end{equation}
After reducing the nine denominators, the local--local part cancels by
the first two resonance moments.  Substitution of the displayed gaps and
the relation \(x(x+y)=z(z+w)\) gives
\begin{equation}
 M_0=\frac{4C_0}{L},\qquad
 M_1=\frac{8C_0\chi}{L},\qquad
 M_2=6C_0
 \label{eq:Midentities}
\end{equation}
modulo the resonance polynomial.  Substitution gives
\begin{equation}
 \cA_{\mathrm{sp}}
 =-\frac{d(B+c\chi)}{aL}+\frac{3d^2}{4a},
 \label{eq:Asp}
\end{equation}
and \eqref{eq:channel-sum} yields \eqref{eq:Wsp}.  Exchanging the triples
proves the second chamber.

For necessity use
\begin{equation}
 p=\kappa+\rho(-5,-2,7),
 \qquad
 q=\kappa+\rho(-7,2,5),
 \label{eq:affine}
\end{equation}
For \(\rho>0\), \(L=14\rho\) and \(\chi=\kappa\); for \(\rho<0\)
the orientation reverses and \(L=14|\rho|\). If \(d=0\), cancellation
forces \(e=0\). If \(d\ne0\), the independent
\(\kappa/|\rho|\), \(1/|\rho|\), and constant coefficients force
\(c=0\), \(B=0\), and \(e=-d^2/(4a)\). Direct substitution proves
sufficiency.
\end{proof}

The table of the nine channel contributions in each chamber and the unreduced
polynomial certificates behind \eqref{eq:Midentities} are supplied in Online
Resource~1.  They provide a direct exact-rational check of the
short analytic reduction above without interrupting the classification
argument.
\subsection{Quadratic leading symbols}
\label{subsec:spatial-symbols}

Assume \(W_6\ne0\) along \eqref{eq:affine}.  The necessary condition
\eqref{eq:ZS}
forces
\begin{equation}
\begin{aligned}
D_\varphi(\kappa,\rho)
={}&\varphi(\kappa-5\rho)+\varphi(\kappa-2\rho)
+\varphi(\kappa+7\rho)\\
&-\varphi(\kappa-7\rho)-\varphi(\kappa+2\rho)
-\varphi(\kappa+5\rho)=0.
\end{aligned}
\label{eq:Dphi}
\end{equation}
The zeroth, first, and second moments of the base triples agree.  Their
cubic-moment difference is \(420\).  Taylor's theorem for
\(\varphi\in C^3(\mathbb R)\) gives
\begin{equation}
 D_\varphi(\kappa,\rho)
 =70\varphi'''(\kappa)\rho^3+o(\rho^3).
 \label{eq:Taylor}
\end{equation}
Off the global cancellation components, the spatial formula on this family is
\[
W_6(\kappa,\rho)
=\frac{2d(B+c\kappa)}{7a|\rho|}-\frac{3d^2}{a}-12e.
\]
If \(d=0\), being off the local cancellation component means \(e\ne0\),
so this is nonzero for every center. If \(d\ne0\), it is nonzero for all
sufficiently small \(|\rho|\) at every center except possibly one:
simultaneous \(B+c\kappa=0\) and \(e=-d^2/(4a)\) can select at most one
\(\kappa\), unless the whole Calogero--Moser component holds. Therefore
\eqref{eq:Taylor} forces \(\varphi'''(\kappa)=0\) at every nonexceptional
center. Since \(\varphi\in C^3(\mathbb R)\), \(\varphi'''\) is
continuous and also vanishes at the possible exceptional center. Thus it
vanishes everywhere, so
\begin{equation}
 \boxed{\varphi(k)=\beta_0+\beta_1k+\beta_2k^2.}
 \label{eq:quadratic-symbol}
\end{equation}
This is exactly \(\varphi\in\Span\{1,k,k^2\}\).  Since
\(H_2=a\int k^2\abs{\alpha(k)}^2\dd k\),
\begin{equation}
 I_2=\beta_0N+\beta_1P+\frac{\beta_2}{a}H_2.
 \label{eq:I2span}
\end{equation}
The \(C^3\) hypothesis is used because this symmetric family cancels moments
through order two; the first informative Taylor coefficient is the third
derivative.

\subsection{Spatial water-wave points}

The relevant source is Fedele and Dutykh
\cite{FedeleDutykh2011,FedeleDutykh2012Extended}.  Their original generic
equation is transformed by the gauge in their equation~(3.4).  In the
relabelled canonical spatial notation of their equations~(3.7) and~(3.9),
\begin{equation}
\begin{aligned}
v_\xi+\ii av_{\tau\tau}+\ii b\abs v^2v
{}&+c_{\mathrm{sp}}\eps\abs v^2v_\tau
+\ii d_{\mathrm{sp}}\eps v\Hilb[(\abs v^2)_\tau]\\
&+\ii e_{\mathrm{sp}}\eps^2\abs v^4v=0,
\end{aligned}
\label{eq:sp-PDE}
\end{equation}
If \((a_{\rm FD},h_{\rm FD},c_{\rm FD},r_{\rm FD},f_{\rm FD})\) denote the
original coefficients (with \(r_{\rm FD}\) replacing their letter \(e\) to
avoid collision with our sextic coefficient), the gauge map is
\[
b=h_{\rm FD},\qquad c_{\rm sp}=c_{\rm FD}-r_{\rm FD},\qquad
d_{\rm sp}=f_{\rm FD},\qquad
e_{\rm sp}=\frac{c_{\rm FD}r_{\rm FD}-2r_{\rm FD}^2}{4a_{\rm FD}}.
\]
After the evolution reversal \eqref{eq:reversal},
the coefficient map is
\begin{equation}
 B=-b,\qquad c=-\eps c_{\mathrm{sp}},\qquad
 d=\eps d_{\mathrm{sp}},\qquad e=\eps^2e_{\mathrm{sp}}.
 \label{eq:map}
\end{equation}
Fedele--Dutykh's potential-envelope point is
\begin{equation}
 (a,b,c_{\mathrm{sp}},d_{\mathrm{sp}},e_{\mathrm{sp}})
 =(1,1,8,2,0).
 \label{eq:potential-point}
\end{equation}
On the explicit test resonance \((\kappa,\rho)=(0,1)\),
\begin{equation}
 W_{6,\mathrm{pot}}
 =-\frac{4\eps(1+21\eps)}7\ne0
 \qquad(\eps>0).
 \label{eq:potential-W}
\end{equation}
Their gauge-Hamiltonian free-surface point is \((1,1,6,2,2)\), giving
\begin{equation}
 W_{6,\mathrm{surf}}
 =-\frac{4\eps(1+63\eps)}7\ne0.
 \label{eq:surface-W}
\end{equation}
The published spatial Hamiltonian and coefficient choices are from
\cite{FedeleDutykh2011,FedeleDutykh2012Extended}; the underlying asymptotic correction
originates in \cite{Dysthe1979}.

\begin{corollary}[Spatial Hamiltonian Dysthe six-wave obstruction]
\label{cor:spatial-physical-obstruction}
Consider either the Fedele--Dutykh potential-envelope point
\[
 (a,b,c_{\rm sp},d_{\rm sp},e_{\rm sp})=(1,1,8,2,0)
\]
or their gauge-Hamiltonian free-surface point
\[
 (a,b,c_{\rm sp},d_{\rm sp},e_{\rm sp})=(1,1,6,2,2),
\]
with physical steepness \(\eps>0\), in the convention
\eqref{eq:sp-PDE} \cite{FedeleDutykh2011,FedeleDutykh2012Extended}.
For every formal conserved functional satisfying the assumptions of
Section~\ref{sec:ZS}, with quadratic leading symbol
\(\varphi\in C^3(\mathbb R)\),
\begin{equation}
 \varphi(k)=\beta_0+\beta_1k+\beta_2k^2.
 \label{eq:spatial-physical-symbol}
\end{equation}
Consequently neither spatial model can support an inverse-scattering
hierarchy with infinitely many linearly independent regular \(C^3\)
quadratic leading symbols.
\end{corollary}

\begin{proof}
The spatial formula \eqref{eq:Wsp}, evaluated along
\eqref{eq:affine}, is nonzero at every center except possibly one; the
continuity argument in Section~\ref{subsec:spatial-symbols} covers that
center.  Equations
\eqref{eq:Dphi}--\eqref{eq:I2span} then give
\eqref{eq:spatial-physical-symbol}.  The symbols \(1\), \(k\), and \(k^2\)
correspond to action, momentum, and the quadratic part of the Hamiltonian.
An inverse-scattering hierarchy with infinitely many independent regular
quadratic leading symbols would require symbols outside this three-dimensional
span.
\end{proof}

\section{Cubic-dispersive temporal slice}
\label{sec:temporal}

Let \(\gamma\ne0\).  In physical Dysthe scaling set
\begin{equation}
 \gamma=\eps h,\qquad c=\eps C,\qquad d=\eps D,
 \qquad e=\eps^2E.
 \label{eq:scaling}
\end{equation}

\subsection{Cubic dispersion creates a genuine four-wave plane}
\label{subsec:four-wave-plane}

For a momentum-conserving quartet put
\[
 S=k_1+k_2=k_3+k_4,
 \quad P_{\mathrm{in}}=k_1k_2,
 \quad P_{\mathrm{out}}=k_3k_4.
\]
Using the second and third Newton sums gives
\begin{equation}
 \Delta_4=-(2a+3\gamma S)(P_{\mathrm{in}}-P_{\mathrm{out}}).
 \label{eq:factor4}
\end{equation}
There are therefore the action-only component and the genuine plane
\begin{equation}
 \boxed{S=-\frac{2a}{3\gamma}.}
 \label{eq:plane}
\end{equation}
The minus sign follows from
\(\omega_\gamma=ak^2+\gamma k^3\).  On the plane write
\(k_{1,2}=S/2\pm r\) and \(k_{3,4}=S/2\pm s\).  Then
\begin{equation}
 T_{\mathrm{plane}}
 =B-\frac{ac}{3\gamma}
 -\frac d2\bigl(\abs{r-s}+\abs{r+s}\bigr).
 \label{eq:Tplane}
\end{equation}
It vanishes for all \(r,s\) if and only if
\begin{equation}
 d=0,\qquad 3\gamma B=ac.
 \label{eq:W4condition}
\end{equation}
This is an order-four obstruction, earlier than the temporal six-wave test.

\subsection{Exact temporal \texorpdfstring{\(3\to3\)}{3-to-3} geometry}

Let \(p,q\) have equal total momentum \(S\) and equal total
\(\omega_\gamma\)-frequency.  Define
\[
 e_2(p)=p_1p_2+p_1p_3+p_2p_3,\qquad
 e_3(p)=p_1p_2p_3,
\]
and similarly for \(q\).  Put
\begin{equation}
\begin{gathered}
 P(z)=\prod_i(z-p_i),\qquad Q(z)=\prod_j(z-q_j),\\
 \Lambda=e_2(p)-e_2(q),\qquad
 \xi=S+\frac{2a}{3\gamma}.
\end{gathered}
\label{eq:temporal-defs}
\end{equation}
The frequency constraint is equivalent to
\(e_3(p)-e_3(q)=\xi\Lambda\), hence
\begin{equation}
 \boxed{P(z)-Q(z)=\Lambda(z-\xi).}
 \label{eq:linear-difference}
\end{equation}
For channel \((\ell,\alpha)\),
\begin{equation}
 \Delta_{\ell\alpha}
 =(q_\alpha-p_i)(q_\alpha-p_j)
 \bigl[2a+3\gamma(S-p_\ell)\bigr],
 \label{eq:temp-den1}
\end{equation}
or equivalently
\begin{equation}
 \boxed{
 \Delta_{\ell\alpha}
 =3\gamma\Lambda
 \frac{(\xi-p_\ell)(q_\alpha-\xi)}{q_\alpha-p_\ell}.
 }
 \label{eq:temp-den2}
\end{equation}
This displays all internal poles explicitly.

Define
\begin{align}
 V^{(1)}_{\alpha\ell}
 &=B+\frac c2(S-p_\ell)
 -\frac d2\sum_{m\ne\ell}\abs{q_\alpha-p_m},
 \label{eq:V1}\\
 V^{(2)}_{\alpha\ell}
 &=B+\frac c2(S-q_\alpha)
 -\frac d2\sum_{n\ne\alpha}\abs{p_\ell-q_n}.
 \label{eq:V2}
\end{align}
The global absolute-value formula is
\begin{equation}
 \boxed{
 \cA_{\mathrm{temp}}(p;q)
 =\sum_{\ell=1}^3\sum_{\alpha=1}^3
 \frac{V^{(1)}_{\alpha\ell}V^{(2)}_{\alpha\ell}}
 {\Delta_{\ell\alpha}},
 \qquad
 W_{6,\mathrm{temp}}=-4\cA_{\mathrm{temp}}-12e.
 }
 \label{eq:temp-kernel}
\end{equation}
After sorting the two triples, the relative order of
\(p_1,p_2,p_3,q_1,q_2,q_3,\xi\) fixes every absolute-value sign.
Formula~\eqref{eq:temp-kernel} is valid before those signs are resolved; the
complete chamber enumeration is supplied in Online Resource~1.

\subsection{Temporal cancellation theorem}

\begin{theorem}[Temporal six-wave cancellation component]
\label{thm:temporal}
For \(a\gamma\ne0\), the temporal six-wave kernel vanishes on every generic
disjoint exact sextet with nonzero virtual denominators if and only if
\begin{equation}
 \boxed{d=0,\qquad e=0,\qquad 3\gamma B=ac.}
 \label{eq:Hirota}
\end{equation}
On this component the four-wave kernel also vanishes on every
non-permutation exact quartet.  Hence \eqref{eq:Hirota} is, in particular,
the unique component of simultaneous non-permutation four-wave and generic
six-wave cancellation.
This is the local Hamiltonian Hirota component, a linear combination of the
NLS and complex modified Korteweg--de Vries flows.  There is no
nonlocal temporal cancellation component; the explicit coefficient matching
and integrability citation are in
Corollary~\ref{cor:integrable-branches}.
\end{theorem}

\begin{proof}
Shift \(k=y-a/(3\gamma)\).  Affine terms in the dispersion do not affect
resonance, and the dispersion becomes \(\gamma y^3\) up to affine terms.
The constant part of the vertex becomes \(B_0=B-ac/(3\gamma)\).
Writing \(E_0=\gamma e\), multiplication of the six-wave equation by
\(\gamma\ne0\) gives the normalized witness equations.
For the six explicit rational pure-cubic
resonances listed with every denominator in
Appendix~\ref{app:temporal-certificate}, clearing only verified nonzero
denominators gives a linear system for the six monomials
\begin{equation}
\bigl(B_0^2,B_0c,B_0d,cd,d^2,E_0\bigr)
\label{eq:temporal-monomial-list}
\end{equation}
whose coefficient matrix has determinant \(59584377600\ne0\).  If the
six-wave coefficient vanishes identically, it vanishes at these witnesses,
so every monomial in \eqref{eq:temporal-monomial-list} is zero.  In
particular \(B_0^2=d^2=E_0=0\), and hence \(B_0=d=e=0\).  This proves
necessity.  Appendix~\ref{app:temporal-certificate} prints the system and its
equivalent ideal certificate.

For sufficiency, if \(B_0=d=e=0\), then
\(T=c(y_1+y_2)/2\).  The channel sum reduces to
\begin{equation}
 \sum_{\ell=1}^3
 \frac1{\prod_{m\ne\ell}(z-p_m)}
 =\frac{3z-S}{P(z)}
 \label{eq:residue}
\end{equation}
and \(P(q_\alpha)=\Lambda(q_\alpha-S)\).  More explicitly,
\[
 \cA_{\rm temp}
 =\frac{c^2}{12\gamma}
   \sum_{\alpha=1}^3\frac{(S-q_\alpha)(3q_\alpha-S)}{P(q_\alpha)}
 =-\frac{c^2}{12\gamma\Lambda}
    \sum_{\alpha=1}^3(3q_\alpha-S)=0.
\]
The bare sextic term is zero, which proves six-wave sufficiency.
Equation~\eqref{eq:W4condition} then gives four-wave cancellation on the same
component.
\end{proof}

\begin{corollary}[Known integrable representatives]
\label{cor:integrable-branches}
The cancellation loci in Theorems~\ref{thm:spatial} and
\ref{thm:temporal} recover the following known integrable equations.
\begin{enumerate}
\item On the spatial local component \(d=e=0\), the case \(c=0\) is cubic
NLS.  If \(c\ne0\), the carrier/Galilean change
\begin{equation}
 u(t,x)=e^{\ii(\nu x-a\nu^2t)}
 q\bigl(t,x-2a\nu t\bigr),\qquad \nu=-\frac{B}{c},
 \label{eq:CLL-carrier}
\end{equation}
converts
\[
 \ii u_t+a u_{xx}-B|u|^2u+\ii c|u|^2u_x=0
\]
into
\[
 \ii q_t+a q_{xx}+\ii c|q|^2q_x=0.
\]
After scaling, this is the Chen--Lee--Liu derivative nonlinear Schr\"odinger
equation.  Thus the whole local component is the mixed Chen--Lee--Liu
derivative nonlinear Schr\"odinger family
\cite{ZakharovManakov1974,ChenLeeLiu1979,Kundu1984}.

\item On the spatial nonlocal component
\[
 B=c=0,\qquad e=-\frac{d^2}{4a},
\]
the equation is
\[
 \ii u_t+a u_{xx}+d\,u|D|(|u|^2)
 -\frac{d^2}{4a}|u|^4u=0.
\]
When \(d/a>0\), scaling gives
\[
 \ii q_t+q_{xx}+q|D|(|q|^2)-\frac14|q|^4q=0,
\]
the focusing canonical gauge representative of the Calogero--Moser
derivative nonlinear Schr\"odinger equation in Appendix~C, equation~(C.5), of G\'erard and
Lenzmann \cite{GerardLenzmann2024}.  For \(d/a<0\), the corresponding
nonlinear phase gauge identifies the equation with the defocusing model
treated in \cite{Badreddine2024}.  The cited integrability results apply on
the Hardy phase spaces of the ungauged equations and their images under
this gauge transformation.  The intermediate-NLS spectral transform of
\cite{PelinovskyGrimshaw1995} provides an earlier integrable nonlocal
lineage; the normalization and phase-space identifications used here
follow \cite{GerardLenzmann2024,Badreddine2024}.

\item On the temporal component \(d=e=0\), \(3\gamma B=ac\), put
\(\alpha=a\), \(\beta=-\gamma\), and
\(\sigma=-B/(2a)\).  Then the equation is exactly
\[
 \ii q_t+\alpha(q_{xx}+2\sigma|q|^2q)
 +\ii\beta(q_{xxx}+6\sigma|q|^2q_x)=0,
\]
the integrable Hirota equation, equivalently a linear combination of the NLS
and complex modified Korteweg--de Vries flows in the
Ablowitz--Kaup--Newell--Segur (AKNS) hierarchy
\cite{Hirota1973,AblowitzKaupNewellSegur1974}.
\end{enumerate}

For every item above, the six-wave coefficient vanishes identically on generic
disjoint exact sextets with nonzero virtual denominators.  On the temporal
Hirota component, the four-wave coefficient also vanishes on every
non-permutation exact quartet.  These statements follow from
Theorems~\ref{thm:spatial} and \ref{thm:temporal}.  Integrability is supplied
independently by the cited
Lax-pair or inverse-scattering literature; vanishing of the present
necessary obstruction is not used as a sufficiency argument.
For the nonlocal component, the integrability attribution carries the
Hardy phase-space restriction just stated, whereas the cancellation
formula in Theorem~\ref{thm:spatial} allows unrestricted real spectral modes.
\end{corollary}

\subsection{Collision symbols for the temporal dispersion}

To classify quadratic leading symbols, we use exact resonant sextets whose
six modes approach an arbitrary center \(\kappa\).  Their virtual denominators
are \(O(\rho^2)\), so the compact-patch estimates do not apply.  The following
collision expansion supplies the required limit.

\begin{theorem}[Temporal quadratic leading symbols]
\label{thm:temporal-symbol}
Let \(a\gamma\ne0\), and let \(I\) satisfy the assumptions of
Section~\ref{sec:ZS}, with quadratic leading symbol
\(\varphi\in C^3(J)\) on a connected interval.  If the coefficients do not
lie on
\[
d=e=0,\qquad ac=3\gamma B,
\]
then every regular additive symbol of a quadratic leading term is
\begin{equation}
\boxed{\varphi(k)=\beta_0+\beta_1k+\beta_2\omega_\gamma(k).}
\label{eq:temporal-symbol-classification}
\end{equation}
\end{theorem}

\begin{proof}
By a \emph{center} we mean a wavenumber \(\kappa\in J\) that is the common
limit of all six external modes in the collision family constructed below.
Fix such a \(\kappa\) and put
\[
A_\kappa=a+3\gamma\kappa,\qquad U_\kappa=B+c\kappa.
\]
For \(A_\kappa\ne0\) and \(\rho>0\), use
\[
p_\rho=\kappa+\rho(-5,-2,7),\qquad
q_\rho=\kappa+\rho(-7,x_{\kappa,\rho},7-x_{\kappa,\rho}).
\]
Exact resonance reduces to
\[
-2A_\kappa(x-2)(x-5)
+21\gamma\rho(10+7x-x^2)=0.
\]
The implicit-function theorem gives the unique branch near \(2\),
\begin{equation}
x_{\kappa,\rho}
=2-\frac{70\gamma}{A_\kappa}\rho+O(\rho^2).
\label{eq:diagonal-x}
\end{equation}
For all sufficiently small \(\rho>0\), the sextet is disjoint and
non-action-only, its chamber is fixed, and its virtual denominators are
nonzero.  Indeed, after \(k=\kappa+\rho y\), they are \(\rho^2\) times the
selected denominators with effective quadratic coefficient \(A_\kappa\).

The nine-channel expansion is derived in
Appendix~\ref{app:diagonal-expansion}; Online Resource~1 independently
reconstructs it in exact characteristic-zero arithmetic.  It gives
\begin{equation}
\boxed{
\begin{aligned}
W_6(p_\rho;q_\rho)
={}&\frac{2dU_\kappa}{7A_\kappa\rho}
+\frac{9\gamma U_\kappa(ac-3\gamma B)}{2A_\kappa^3}\\
&-\frac{3d^2}{A_\kappa}-12e+O(\rho).
\end{aligned}}
\label{eq:diagonal-W}
\end{equation}
The same branch and Taylor's theorem give
\begin{equation}
\boxed{
D_\varphi(p_\rho;q_\rho)
=70\rho^3\left[
\varphi'''(\kappa)-\frac{3\gamma}{A_\kappa}\varphi''(\kappa)
\right]+o(\rho^3).}
\label{eq:diagonal-Dphi}
\end{equation}

Outside the cancellation set, the expansion is nonzero for every
sufficiently small positive \(\rho\) at all but finitely many centers
\(\kappa\).  Indeed, its coefficients are rational functions of \(\kappa\),
with denominators that are powers of the affine function \(A_\kappa\).
If \(d\ne0\) and \(U_\kappa\not\equiv0\), the \(1/\rho\)
coefficient is nonzero away from the at most one zero of the affine function
\(U_\kappa\) and the possible zero of \(A_\kappa\).  If \(d\ne0\) and
\(U_\kappa\equiv0\), the finite coefficient
\(-3d^2/A_\kappa-12e\) is not identically zero and, after multiplication by
\(A_\kappa\), has a nonzero affine numerator.  If \(d=0\), the finite
coefficient
\[
\frac{9\gamma U_\kappa(ac-3\gamma B)}{2A_\kappa^3}-12e
\]
can vanish identically only when \(e=0\) and \(ac=3\gamma B\); otherwise,
after multiplication by \(A_\kappa^3\), its numerator is a nonzero polynomial
and therefore has only finitely many real zeros.  At every remaining center,
the first displayed nonzero Laurent coefficient ensures
\(W_6(p_\rho;q_\rho)\ne0\) for all sufficiently small positive \(\rho\).
The only possible exceptions are the finitely many zeros just described and
the possible point \(A_\kappa=0\).

At every remaining center, condition~\eqref{eq:ZS} and
\eqref{eq:diagonal-Dphi} give
\[
G(\kappa):=A_\kappa\varphi'''(\kappa)-3\gamma\varphi''(\kappa)=0.
\]
Because \(\varphi\in C^3\), \(G\) is continuous.  It therefore vanishes at
the finitely many exceptional centers and at the possible point
\(A_\kappa=0\) as well.  Off that point,
\[
\left(\frac{\varphi''}{A_\kappa}\right)'=0.
\]
Continuity of \(\varphi'''\) matches the constants across
\(A_\kappa=0\), so
\(\varphi''=c_*(a+3\gamma k)\) for one constant \(c_*\) on the connected interval.
Integrating twice proves \eqref{eq:temporal-symbol-classification}.
\end{proof}

For \(\omega_\gamma(k)=ak^2+\gamma k^3\), the functions in
\eqref{eq:temporal-symbol-classification} are exactly the symbols of \(N\),
\(P\), and \(H_2\). At \(\gamma=0\), the separate spatial theorem gives the span
\(\{1,k,k^2\}\).

\subsection{Craig--Guyenne--Sulem temporal coefficient map}

Craig, Guyenne, and Sulem derive a canonical temporal Hamiltonian Dysthe
equation in their equations~(17)--(21) and simplify it in Section~5.2 after
subtracting wave action and impulse \cite{CraigGuyenneSulem2021}.  With slow
time \(\tau=\eps^2t\) and effective Hamiltonian
\(K=\eps^{-3}\widehat H\), their symplectic equation~(19) becomes
\(\ii u_\tau=\delta K/\delta\bar u\), exactly the convention
\eqref{eq:canonical}.  Comparing the Hamiltonian terms gives
\begin{equation}
\boxed{
a=\frac{\omega_0''}{2},\quad
\gamma=\frac{\eps\omega_0'''}6,\quad
B=k_0^3,\quad c=3\eps k_0^2,\quad
d=\eps k_0^2,\quad e=0.
}
\label{eq:CGS-map}
\end{equation}
For deep water, \(\omega(k)=\sqrt{g_{\rm grav}k}\).  Hence
\[
a=-\frac{g_{\rm grav}^2}{8\omega_0^3},\qquad
\gamma=\frac{\eps g_{\rm grav}^3}{16\omega_0^5}.
\]
In the nondimensional normalization
\(g_{\rm grav}=k_0=\omega_0=1\),
\begin{equation}
\boxed{(a,\gamma,B,c,d,e)
=\left(-\frac18,\frac\eps{16},1,3\eps,\eps,0\right).}
\label{eq:CGS-point}
\end{equation}
The PDE obtained from \eqref{eq:PDE} agrees term by term with their
Section~5.2 equation.  This identifies the canonical
Craig--Guyenne--Sulem branch with linear dispersion truncated at cubic order.
Since
\(d=\eps\ne0\), its coefficient tuple lies outside the simultaneous temporal
cancellation component.

\begin{corollary}[Cubic-dispersion temporal Hamiltonian Dysthe six-wave
obstruction]
\label{cor:temporal-physical-obstruction}
Consider the Craig--Guyenne--Sulem temporal Hamiltonian Dysthe model with
linear dispersion truncated at cubic order and \(\eps\ne0\), whose normalized
coefficients are \cite{CraigGuyenneSulem2021}
\[
 (a,\gamma,B,c,d,e)
 =\left(-\frac18,\frac{\eps}{16},1,3\eps,\eps,0\right).
\]
For every formal conserved functional satisfying the assumptions of
Section~\ref{sec:ZS}, with quadratic leading symbol
\(\varphi\in C^3(J)\) on the connected spectral interval \(J\),
\begin{equation}
 \varphi(k)=\beta_0+\beta_1k+\beta_2\omega_\gamma(k).
 \label{eq:temporal-physical-symbol}
\end{equation}
Consequently this model cannot support an inverse-scattering hierarchy with
infinitely many linearly independent regular \(C^3\) quadratic leading
symbols.
\end{corollary}

\begin{proof}
Here \(d=\eps\ne0\), so the tuple lies outside the unique simultaneous
four- and six-wave cancellation component of
Theorem~\ref{thm:temporal}.  The diagonal blow-up theorem
\ref{thm:temporal-symbol} gives
\eqref{eq:temporal-physical-symbol}; the three terms are the symbols of
action, momentum, and the quadratic part of the Hamiltonian.  The
inverse-scattering conclusion follows as in
Corollary~\ref{cor:spatial-physical-obstruction}.
\end{proof}

\section{Uniform deformation of resonance patches}
\label{sec:deformation}

\subsection{A deformation parameter distinct from steepness}

Introduce \(\lambda\) by
\begin{equation}
 \omega_{\lambda,\eps}(k)=ak^2+\lambda\eps hk^3,
 \label{eq:lambda-dispersion}
\end{equation}
while retaining
\begin{equation}
 c=\eps C,\qquad d=\eps D,\qquad e=\eps^2E.
 \label{eq:Dysthe-scaling}
\end{equation}
The spatial slice is \(\lambda=0\); the cubic-dispersion temporal slice is
\(\lambda=1\).  The limit \(\lambda\to0\) turns off cubic dispersion only.
The physical limit \(\eps\to0\) turns off all higher-order Dysthe corrections
together.

\subsection{Proof of uniform persistence on resonance patches}

\begin{proof}[Proof of Theorem~\ref{thm:strong-open}]
For \(z=(p,q)\in U\), retain the notation
\[
P_z(X)=\prod_{i=1}^3(X-p_i),\qquad
Q_z(X)=\prod_{\alpha=1}^3(X-q_\alpha),\qquad
P_z-Q_z=C_0(z).
\]
Put
\[
\eta=\lambda\eps h,\qquad
\tau_\eta(z)=\frac{3\eta}{2a+3\eta S(z)},
\qquad
Q_{\eta,z}(X)=Q_z(X)+C_0(z)\tau_\eta(z)X.
\]
For \(|\eta|\) small, the three simple real roots of \(Q_z\) continue to
simple real analytic roots \(q_\alpha(\eta,z)\), uniformly on \(K\).  Define
\[
\Phi_{\lambda,\eps}(p,q)
=\bigl(p,q(\lambda\eps h;p,q)\bigr).
\]

The \(X^2\)-coefficient of \(Q_{\eta,z}\) is unchanged, so the outgoing
momentum remains \(S\).  If
\(\Lambda=e_2(p)-e_2(q_\eta)=-C_0\tau_\eta\), its constant term gives
\(e_3(p)-e_3(q_\eta)=-C_0\).  Newton's identities therefore yield
\[
\sum p_i^2-\sum q_{\eta,\alpha}^2=-2\Lambda,
\qquad
\sum p_i^3-\sum q_{\eta,\alpha}^3=-3S\Lambda-3C_0.
\]
Consequently the frequency difference for
\(\omega_\eta(k)=ak^2+\eta k^3\) is
\[
C_0\bigl[(2a+3\eta S)\tau_\eta-3\eta\bigr]=0.
\]
Thus every image point is an exact resonance.  Root separation and the
strict inequalities defining \(U\) persist uniformly on \(K\), proving the
fixed-chamber and distinct-slot statements, together with a uniform positive
lower bound for the absolute values of the internal denominators.  The map is
injective:
the incoming polynomial is unchanged, \(C_0=P(0)-Q_\eta(0)\) is recovered
from the image, and then \(Q=Q_\eta-C_0\tau_\eta X\).  Hence the map is an
analytic embedding: the displayed recovery formula defines an analytic
inverse on a relative neighborhood of the image.  Since the source and target
resonance manifolds have equal dimension, it is a local diffeomorphism onto
its relative image.

It remains to calculate the first variation.  At \(\eps=0\),
\eqref{eq:sp-den} holds.  Differentiating the two vertices, with the
cross-difference sums defined in \eqref{eq:spatial-cross-sums}, gives
\begin{equation}
\dot T^{(1)}_{\ell\alpha}+\dot T^{(2)}_{\ell\alpha}
=C\left(S-\frac{p_\ell+q_\alpha}{2}\right)
-\frac D2\left(R^{(1)}_{\ell\alpha}+R^{(2)}_{\ell\alpha}\right).
\label{eq:open-vertex-variation}
\end{equation}
The contribution from the local derivative nonlinearity cancels by the two
resonance moments:
\[
\sum_{\ell,\alpha}(q_\alpha-p_\ell)
\left(S-\frac{p_\ell+q_\alpha}{2}\right)=0.
\]
The chamber identity \(M_0=4C_0/L\) then gives
\[
\left.\partial_\eps\cA\right|_{\rm mean\ flow}
=-\frac{BD}{aL}.
\]

The remaining possible first-order contribution comes from the constant
vertex \(B\), cubic dispersion, and resonance motion.  Along the explicit
embedding,
\[
P(q_{\eta,\alpha})=C_0(1-\tau_\eta q_{\eta,\alpha})
\]
and the exact channel factorization is
\[
\Delta_{\ell\alpha}(\eta)
=(q_{\eta,\alpha}-p_i)(q_{\eta,\alpha}-p_j)
\bigl[2a+3\eta(S-p_\ell)\bigr].
\]
A partial-fraction sum gives, for \(\eta\ne0\), the exact identity
\begin{equation}
\boxed{
\sum_{\ell,\alpha}\frac1{\Delta_{\ell\alpha}(\eta)}
=\frac1{\eta P(1/\tau_\eta)}.
}
\label{eq:local-local-open-identity}
\end{equation}
For completeness, put \(z=1/\tau_\eta\).  Then
\(2a+3\eta(S-p_\ell)=3\eta(z-p_\ell)\) and
\[
(q-p_i)(q-p_j)=\frac{P(q)}{q-p_\ell}
=\frac{C_0(z-q)}{z(q-p_\ell)}.
\]
Consequently the double sum reduces to
\[
\frac{z}{3\eta C_0}
\left[3\frac{Q_\eta'(z)}{Q_\eta(z)}
-3\frac{P'(z)}{P(z)}\right].
\]
The identities \(Q_\eta(z)=Q(z)+C_0=P(z)\) and
\(Q_\eta'(z)=P'(z)+C_0/z\) give
\eqref{eq:local-local-open-identity}.
Since \(P\) is monic,
\[
 P(1/\tau_\eta)
 =\tau_\eta^{-3}
 \left(1-S\tau_\eta+e_2(p)\tau_\eta^2-e_3(p)\tau_\eta^3\right)
 =\tau_\eta^{-3}\bigl(1+O_K(\tau_\eta)\bigr).
\]
The parenthetical factor is uniformly bounded away from zero for small
\(\eta\), and \(\tau_\eta=O_K(\eta)\).  Therefore the right-hand side of
\eqref{eq:local-local-open-identity} is \(O_K(\eta^2)\) and extends
continuously by zero at \(\eta=0\).
Thus the \(B^2\) local--local part has no first variation.
Combining this with \eqref{eq:channel-sum}, while observing that the bare
sextic is \(O(\eps^2)\), proves \eqref{eq:open-patch-expansion}.

For part~\textup{(ii)}, the construction above is analytic on a relative
neighborhood of \(K=\overline V\), and the explicit inverse
\(Q=Q_\eta-C_0\tau_\eta X\) shows that its restriction to \(V\) is a local
diffeomorphism between equal-dimensional resonance manifolds.  Therefore
\(\Phi_{\lambda,\eps}(V)\) is relatively open in
\(\mathcal R_{\lambda,\eps}\).  The uniform estimates on \(K\) supply the
same chamber and denominator bounds.  If \(BD\ne0\), continuity and
compactness give a positive lower bound for \(|4BD/(aL)|\) on \(K\); the
uniform remainder then yields the stated \(c_K|\eps|\) estimate.
\end{proof}

\paragraph{Example.}
For \(p=(-5,-2,7)\) and
\(q_{\lambda,\eps}=(-7,x_{\lambda,\eps},7-x_{\lambda,\eps})\), exact
resonance reduces to
\[
 -2a(x-2)(x-5)+21\lambda\eps h(10+7x-x^2)=0.
\]
The branch through \(x=2\) satisfies
\[
 x_{\lambda,\eps}
 =2-\frac{70\lambda h}{a}\eps+O(\eps^2).
\]
Since \(L=14\), Theorem~\ref{thm:strong-open} gives
\[
 W_{6,\lambda,\eps}(p;q_{\lambda,\eps})
 =\frac{2BD}{7a}\eps+O(\eps^2).
\]

\begin{remark}[Independence of the resonance parametrization]
\label{rem:intrinsic-first-variation}
Any two analytic identifications of the deformed and spatial resonance
manifolds differ at first order by a vector tangent to \(\mathcal R_0\).
Since the unperturbed constant-vertex coefficient vanishes identically on
\(\mathcal R_0\), its tangential derivative is zero.  Thus the first-variation
function in \eqref{eq:open-patch-expansion} is intrinsic, although the
particular polynomial embedding is chosen for transparency.
\end{remark}

\section{Uniform narrowband matching with exact deep-water Euler}
\label{sec:euler-matching}

We now compute the complete degree-six coefficient of the formal Euler normal
form after cubic and quartic elimination and compare it uniformly with the
temporal Dysthe coefficient on compact resonance patches.  The calculation
uses only the finite normal-form expansion through degree six; no convergent
infinite normal-form transformation is assumed.

Throughout this section, \(\mu>0\) denotes relative bandwidth.  It is distinct
from the general parameter \(\eps\) used in
Theorem~\ref{thm:strong-open}; the temporal Dysthe comparison is obtained by
setting \(\lambda=1\) and \(\eps=\mu\) only within this section.

\subsection{Exact Euler resonance blow-up and the matching theorem}

In units \(g=k_0=1\), the infinite-depth irrotational water-wave Hamiltonian
in Zakharov surface variables is
\begin{equation}
 \mathcal H^{\rm E}(\eta,\xi)
 =\frac12\int_{\mathbb R}
 \left[\eta^2+\xi G(\eta)\xi\right]\dd x,
 \label{eq:Euler-Hamiltonian}
\end{equation}
where \(G(\eta)\) is the Dirichlet--Neumann operator
\cite{Zakharov1968,CraigSulem1993}.  Write
\(\mathcal H^{\rm E}=H_2+H_3+H_4+\cdots\).  To avoid collision with the
dispersion coefficient \(a\), denote the canonical Euler normal variable by
\(b(k)\).  With \(s(k)=|k|^{-1/4}\),
\begin{equation}
 \widehat\eta(k)=\frac{b(k)+\overline{b(-k)}}{\sqrt2\,s(k)},
 \qquad
 \widehat\xi(k)=-\frac{\ii s(k)}{\sqrt2}
 \left[b(k)-\overline{b(-k)}\right].
 \label{eq:Euler-normal-map}
\end{equation}
We use the Fourier measures \eqref{eq:Gamma4} and \eqref{eq:Gamma6}, the
bracket \eqref{eq:bracket} with \(\alpha\) replaced by \(b\), and the ordered
sextic convention \eqref{eq:H6}.  Thus \(W_6^{\rm E}\) is defined by
\begin{equation}
 H_{6,{\rm NF}}^{\rm E}
 =\frac1{3!3!}\int W_6^{\rm E}(p;q)
 \prod_{j=1}^3\overline{b(q_j)}
 \prod_{i=1}^3b(p_i)\dd\Gamma_6,
 \label{eq:Euler-W6-definition}
\end{equation}
after elimination of the nonresonant cubic and then nonresonant quartic
Hamiltonians.  We call this the complete connected Euler coefficient: it
includes the direct sextic vertex, every connected cubic-elimination term, and
the quartic normal-form contractions contributing to the ordered disjoint
six-leg target.  These are the classical canonical normal-form conventions of
\cite{Krasitskii1994,CraigWorfolk1995,CraigSulem2016}.

Remove the carrier and group-velocity terms from the exact dispersion relation
restricted to positive wavenumbers and define
\begin{equation}
 \Omega_\mu^{\rm E}(\kappa)
 =\frac{\sqrt{1+\mu\kappa}-1-\mu\kappa/2}{\mu^2},
 \qquad
 \Omega_0^{\rm E}(\kappa)=-\frac18\kappa^2.
 \label{eq:Euler-rescaled-dispersion}
\end{equation}
Taylor's theorem gives, uniformly for \(\kappa\) in a compact set,
\begin{equation}
 \Omega_\mu^{\rm E}(\kappa)
 =-\frac18\kappa^2+\frac\mu{16}\kappa^3
 -\frac{5\mu^2}{128}\kappa^4+O_K(\mu^3).
 \label{eq:Euler-dispersion-expansion}
\end{equation}
The first two terms are the temporal Dysthe dispersion with
\(a=-1/8\), \(h=1/16\).

\begin{theorem}[Uniform Euler--Dysthe six-wave matching]
\label{thm:euler-dysthe-matching}
Let \(U\subset\mathcal R_0^{\rm gen}\) be a connected fixed chamber and let
\(K\Subset U\).  At the coefficient point
\eqref{eq:intro-CGS-source}, there are \(\mu_K>0\) and restrictions to \(K\)
of real-analytic resonance identifications
\[
 \Psi_\mu:K\longrightarrow\widetilde{\mathcal R}^{\rm E}_\mu,
 \qquad
 \Phi_\mu:K\longrightarrow\mathcal R_{1,\mu},
 \qquad 0\le\mu<\mu_K,
\]
defined on a relative neighborhood of \(K\), with
\(\Psi_0=\Phi_0=\operatorname{id}\) and
\(\Psi_\mu-\Phi_\mu=O_K(\mu^2)\).  Here
\begin{equation}
 \widetilde{\mathcal R}^{\rm E}_\mu
 =\left\{(p,q):\sum p_i=\sum q_\alpha,\quad
 \sum_i\Omega_\mu^{\rm E}(p_i)
 =\sum_\alpha\Omega_\mu^{\rm E}(q_\alpha)\right\}.
 \label{eq:Euler-resonance-patch}
\end{equation}
For \(\mu>0\), the physical wavenumbers
\(1+\mu\Psi_\mu(z)_j\) are positive and distinct, the chamber is unchanged,
and all nine rescaled virtual denominators are uniformly bounded away from
zero.  If \(\widehat W_6^{\rm D}\) is the ordered temporal Dysthe coefficient,
then
\begin{equation}
 \boxed{
 W_6^{\rm E}\bigl(1+\mu\Psi_\mu(z)\bigr)
 =\mu^{-2}\widehat W_6^{\rm D}
   \bigl(\Phi_\mu(z)\bigr)+O_K(1)
 }
 \label{eq:Euler-Dysthe-strong-match}
\end{equation}
and
\begin{equation}
 \boxed{
 \frac\mu4W_6^{\rm E}\bigl(1+\mu\Psi_\mu(z)\bigr)
 =\frac1{4\mu}\widehat W_6^{\rm D}
   \bigl(\Phi_\mu(z)\bigr)+O_K(\mu)
 =-\frac8{L(z)}+O_K(\mu).
 }
 \label{eq:Euler-Dysthe-normalized-match}
\end{equation}
All remainders are uniform for \(z\in K\).
\end{theorem}

At the coefficient point in the theorem,
\(BD/(aL)=-8/L\).

\begin{lemma}[Uniform exact-resonance embedding]
\label{lem:Euler-embedding}
The resonance identifications in
Theorem~\ref{thm:euler-dysthe-matching} exist and have the stated agreement
through first order, together with the chamber and denominator properties.
\end{lemma}

\begin{proof}
Fix the incoming triple and one outgoing coordinate, use momentum to remove a
second outgoing coordinate, and solve the frequency constraint for the third.
At \(\mu=0\), differentiation with respect to that coordinate produces a
nonzero multiple of the difference between two distinct outgoing modes.
This is bounded away from zero on \(K\).  The analytic implicit-function
theorem, applied on finitely many neighborhoods covering \(K\), gives a
uniform analytic Euler embedding.  The same construction applied to
\(-\kappa^2/8+\mu\kappa^3/16\) gives \(\Phi_\mu\).  Expansion
\eqref{eq:Euler-dispersion-expansion} shows that the two implicit equations
agree through first order; uniqueness therefore gives
\(\Psi_\mu-\Phi_\mu=O_K(\mu^2)\).  Every chamber inequality and every
quadratic virtual denominator is strict on \(K\), so compactness preserves
the signs and a common positive denominator bound for sufficiently small
\(\mu\).  Positivity and distinctness of the physical wavenumbers follow in
the same way.
\end{proof}

\subsection{Uniform matching of the quartic symbols}

Let \(T^{\rm E}\) denote the ordered \(2\to2\) Euler kernel after elimination
of the cubic Hamiltonian.  Craig and Sulem identify the cubic normal-form
generator and its quartic output.  With \(\mathsf D_x=-\ii\partial_x\), in the
present sign convention this is
\begin{equation}
 H_4^{(3)}(\eta,\xi)
 =H_4(\eta,\xi)-H_4(-\ii\operatorname{sgn}(\mathsf D_x)\eta,\xi).
 \label{eq:Euler-transformed-quartic}
\end{equation}
See Craig--Sulem \cite[Theorem~4.1]{CraigSulem2016} and the derivation leading
to equation~(16) of Craig--Guyenne--Sulem
\cite{CraigGuyenneSulem2021}.

For the coefficients \eqref{eq:intro-CGS-source}, write the Dysthe vertex as
\begin{equation}
 T_\mu^{\rm D}(k_1,k_2;k_3,k_4)
 =1+\mu\left[
 \frac32(k_1+k_2)
 -\frac14\sum_{r=1}^2\sum_{s=3}^4|k_r-k_s|
 \right].
 \label{eq:source-Dysthe-vertex}
\end{equation}

\begin{lemma}[Uniform quartic-symbol matching]
\label{lem:Euler-quartic-match}
On every compact fixed chamber of sideband quartets,
\begin{equation}
 T^{\rm E}(1+\mu k_1,1+\mu k_2;1+\mu k_3,1+\mu k_4)
 =T_\mu^{\rm D}(k_1,k_2;k_3,k_4)+O_K(\mu^2).
 \label{eq:Euler-quartic-match}
\end{equation}
The remainder is real analytic after the chamber signs are fixed and is
uniformly bounded together with its required derivatives.
\end{lemma}

\begin{proof}
Polarizing the exact transformed Hamiltonian
\eqref{eq:Euler-transformed-quartic} in the convention
\eqref{eq:Euler-W6-definition} gives, for positive physical wavenumbers
\(K_1+K_2=K_3+K_4\),
\begin{equation}
 \boxed{
 T^{\rm E}(K_1,K_2;K_3,K_4)
 =\frac12 K_{\min}(K_1K_2K_3K_4)^{1/4}
 \left(\sqrt{K_1K_2}+\sqrt{K_3K_4}\right),
 }
 \label{eq:Euler-positive-quartic-kernel}
\end{equation}
where \(K_{\min}=\min_jK_j\).  Formula
\eqref{eq:Euler-positive-quartic-kernel} follows directly from the two terms
in \eqref{eq:Euler-transformed-quartic}; it is the positive-wavenumber
specialization of the cubic-normal-form quartic kernel
\cite{CraigSulem2016,CraigGuyenneSulem2021}.

Set \(K_j=1+\mu k_j\),
\(S=k_1+k_2=k_3+k_4\), and
\(m=\min_jk_j\).  On a fixed chamber the same entry realizes the minimum, and
\begin{align*}
 K_{\min}&=1+\mu m,\\
 (K_1K_2K_3K_4)^{1/4}&=1+\frac{\mu S}{2}+O_K(\mu^2),\\
 \sqrt{K_1K_2}+\sqrt{K_3K_4}&=2+\mu S+O_K(\mu^2).
\end{align*}
Consequently the exact kernel is
\[
 T^{\rm E}=1+\mu(S+m)+O_K(\mu^2).
\]
Equal pair sums imply the elementary chamber identity
\[
 \sum_{r=1}^2\sum_{s=3}^4|k_r-k_s|=2S-4m.
\]
Thus
\[
 S+m=\frac32S-\frac14
 \sum_{r=1}^2\sum_{s=3}^4|k_r-k_s|,
\]
which is exactly the first Taylor coefficient in
\eqref{eq:source-Dysthe-vertex}.  Positivity of the physical wavenumbers and
the fixed minimum on a neighborhood of the compact chamber make
\eqref{eq:Euler-positive-quartic-kernel} real analytic there.  Taylor's
theorem supplies the uniform \(O_K(\mu^2)\) remainder, together with the
required derivative bounds.
\end{proof}

\subsection{The singular contraction and the bounded Euler remainder}

For a channel \((\ell,\alpha)\), use the internal sideband
\(r_{\ell\alpha}\) from \eqref{eq:internal}.  Along \(\Psi_\mu\), define the
rescaled exact Euler divisor
\begin{equation}
 \widehat\Delta^{\rm E}_{\ell\alpha}
 =\mu^{-2}\left[
 \sqrt{1+\mu q_\alpha}+\sqrt{1+\mu r_{\ell\alpha}}
 -\sqrt{1+\mu p_i}-\sqrt{1+\mu p_j}
 \right].
 \label{eq:Euler-rescaled-divisor}
\end{equation}
Let \(\Delta^{\rm D}_{\ell\alpha}\) be the corresponding divisor for the
temporal Dysthe dispersion along \(\Phi_\mu\).  Lemma~\ref{lem:Euler-embedding}
and \eqref{eq:Euler-dispersion-expansion} give
\begin{equation}
 \widehat\Delta^{\rm E}_{\ell\alpha}
 =\Delta^{\rm D}_{\ell\alpha}+O_K(\mu^2).
 \label{eq:Euler-divisor-match}
\end{equation}
The Euler vertices below are evaluated at physical wavenumbers
\(1+\mu(q,r,p_i,p_j)\), whereas the Dysthe vertices use sideband coordinates
before conversion to physical wavenumbers.

\begin{lemma}[Bounded coefficients after cubic elimination]
\label{lem:Euler-symbol-class}
Through homogeneous degree six, every coefficient produced from
\(H_3,\ldots,H_6\) by cubic elimination in the physical variables
\((\eta,\xi)\) is a finite sum of products of Fourier arguments, their
absolute values with nonnegative integer powers, and bounded sign factors.
No negative power of an internal Fourier argument is introduced.  Conversion
to the normal variable \(b\) contributes the factors
\(|k_j|^{\pm1/4}\) only on the external legs.
\end{lemma}

\begin{proof}
In physical variables the cubic generator is
\begin{equation}
 K_3(\eta,\xi)=\frac12\int
 \bigl(-\ii\operatorname{sgn}(\mathsf D_x)\eta\bigr)^2
 |\mathsf D_x|\xi\dd x.
 \label{eq:Euler-cubic-generator}
\end{equation}
The Craig--Sulem recursion for \(H_3,\ldots,H_6\)
\cite{CraigSulem1993} and the generator
\eqref{eq:Euler-cubic-generator} use only multiplication,
\(\mathsf D_x\), \(|\mathsf D_x|\),
\(\operatorname{sgn}(\mathsf D_x)\), and finite composition
\cite{CraigGuyenneSulem2021}.  Variational differentiation and Poisson
contraction in \((\eta,\xi)\) preserve this finite class.  In particular, an
internal \(\eta\)--\(\xi\) contraction is taken before the normal-variable map
\eqref{eq:Euler-normal-map}; it therefore does not acquire a
\(|k|^{-1/4}\) factor.  Only after the physical-variable coefficient has been
formed are its six external fields converted to \(b\), giving the stated
external factors.  This proves the claim by induction over the finite
Craig--Sulem recursion and the finite number of Poisson brackets.
\end{proof}

\begin{lemma}[Only the quartic contraction is singular]
\label{lem:Euler-bounded-remainder}
On the embedded patch,
\begin{equation}
 W_6^{\rm E}
 =-\frac4{\mu^2}\sum_{\ell,\alpha=1}^3
 \frac{T^{\rm E}_{1,\ell\alpha}T^{\rm E}_{2,\ell\alpha}}
 {\widehat\Delta^{\rm E}_{\ell\alpha}}
 +R_6^{\rm E},
 \qquad
 \sup_{z\in K,\,0<\mu<\mu_K}|R_6^{\rm E}|<\infty.
 \label{eq:Euler-singular-decomposition}
\end{equation}
\end{lemma}

\begin{proof}
The quartic monomials with two unbarred and two barred fields preserve wave
action.  Their homological equation gives exactly the nine one-contraction
channels already counted in Section~\ref{sec:normal-form}.
The physical divisor is \(\mu^2\widehat\Delta^{\rm E}\), and the same ordered
factorial calculation gives the coefficient \(-4\).

It remains to exclude any additional singular sector.  Assign carrier sign
\(+1\) to an unbarred external leg and \(-1\) to a barred leg.  For a
non-number-preserving quartic monomial, the signed sum of its three external
carriers is \(S\in\{-3,-1,1,3\}\).  If \(\sigma_0=\pm1\) is the orientation
of the internal leg, leading momentum and frequency mismatch give
\begin{equation}
 r_0=-S/\sigma_0,
 \qquad
 \Delta_0=S+\sigma_0\sqrt{|S|}.
 \label{eq:Euler-carrier-count}
\end{equation}
The only zero occurs for \(|S|=1\) and
\(\sigma_0=-\operatorname{sgn}S\), which is precisely the
number-preserving \(2\to2\) sector already displayed.  Every other value has
absolute size in \(\{2,3-\sqrt3,3+\sqrt3\}\), so compactness preserves a
uniform positive bound.

With the bracket convention \eqref{eq:bracket}, this generator satisfies
\(\{H_2,K_3\}=H_3\).  Cubic elimination therefore uses
\(e^{-\ad_{K_3}}\), corresponding to the source flow at parameter \(-1\).
Its degree-six part is the finite polynomial
\begin{equation}
 H_6^{(3)}=
 \sum_{j=0}^4\frac{(-1)^j}{j!}
 \operatorname{ad}_{K_3}^{j}H_{6-j}.
 \label{eq:Euler-degree-six-Lie}
\end{equation}
Lemma~\ref{lem:Euler-symbol-class} makes every term in
\eqref{eq:Euler-degree-six-Lie} uniformly bounded, including when an internal
wavenumber tends to zero.  Each internal argument is \(m+\mu\ell(z)\), with
\(m\) an integer multiple of the carrier, while the six external factors
\((1+\mu z_j)^{\pm1/4}\) remain bounded on \(K\).

Table~\ref{tab:Euler-sectors} lists every connected degree-six sector after
cubic and quartic normalization.
\begin{table}[htbp]
\centering
\small
\begin{tabularx}{\textwidth}{@{}p{0.25\textwidth}p{0.22\textwidth}Y@{}}
\toprule
Sector & Size on the patch & Reason\\
\midrule
Direct \(H_6\) & \(O_K(1)\) & Lemma~\ref{lem:Euler-symbol-class}\\
\(\ad_{K_3}H_5\) & \(O_K(1)\) & Finite physical-variable contraction with
no inverse internal multiplier\\
\(\ad_{K_3}^2H_4\) & \(O_K(1)\) & Same symbol-class closure\\
\(\ad_{K_3}^3H_3\) and \(\ad_{K_3}^4H_2\) & \(O_K(1)\) & Same closure;
their coefficients combine as prescribed by
\eqref{eq:Euler-degree-six-Lie}\\
\(\{Z_4,F_4\}\) & \(0\) or \(O_K(1)\) & Positive resonant quartets are
action-only; every other relevant
quartic mismatch is uniformly nonzero\\
\(\tfrac12\{P_4,F_4\}\) & Nine channels of order \(\mu^{-2}\);
remainder \(O_K(1)\) & Only the number-preserving \(2\to2\) mismatch vanishes
at carrier order\\
\bottomrule
\end{tabularx}
\caption{Complete connected degree-six Euler sectors and their narrowband
sizes.  Lie coefficients and signs are those of
\eqref{eq:Euler-degree-six-Lie} and the quartic Lie transform.}
\label{tab:Euler-sectors}
\end{table}

A subsequent quintic normalization cannot change this degree-six
coefficient.  For homogeneous functionals,
\(\deg\{A_m,B_n\}=m+n-2\).  The bracket \(\{H_2,F_5\}\) removes the degree-five
term; the only possible degree-six bracket would involve the already
eliminated \(H_3\), while a bracket with \(H_4\) has degree seven.

It remains to justify the action-only entry in the table.  For fixed positive
total wavenumber
\(S\), the function \(f_S(x)=\sqrt{x}+\sqrt{S-x}\) is symmetric about
\(S/2\) and strictly monotone on either side.  Hence
\(f_S(k_1)=f_S(k_3)\) implies
\(\{k_1,S-k_1\}=\{k_3,S-k_3\}\): every exact resonant quartet in the
positive-wavenumber sector is a permutation quartet.  The resonant quartic
contribution from
this positive-wavenumber sector is therefore action-only, and its bracket
with \(F_4\) cannot create a disjoint six-leg target.  Every other quartic
sector that can contract to this target has the uniform nonzero mismatch
established above, so no additional resonant-quartic contribution occurs.
This proves
\eqref{eq:Euler-singular-decomposition}.
\end{proof}

\subsection{Proof of the Euler--Dysthe matching theorem}

\begin{proof}[Proof of Theorem~\ref{thm:euler-dysthe-matching}]
Lemma~\ref{lem:Euler-quartic-match},
\eqref{eq:Euler-divisor-match}, and the common denominator lower bound imply,
channel by channel,
\[
 \frac{T^{\rm E}_{1,\ell\alpha}T^{\rm E}_{2,\ell\alpha}}
 {\widehat\Delta^{\rm E}_{\ell\alpha}}
 =
 \frac{T^{\rm D}_{1,\ell\alpha}T^{\rm D}_{2,\ell\alpha}}
 {\Delta^{\rm D}_{\ell\alpha}}+O_K(\mu^2).
\]
There are nine channels.  Substitution into
\eqref{eq:Euler-singular-decomposition} gives
\[
 W_6^{\rm E}\bigl(1+\mu\Psi_\mu(z)\bigr)
 =\mu^{-2}\widehat W_6^{\rm D}
 \bigl(\Phi_\mu(z)\bigr)+O_K(1),
\]
which is \eqref{eq:Euler-Dysthe-strong-match}.

The general Dysthe resonance-patch theorem has already evaluated the source
coefficient:
\[
 \widehat W_6^{\rm D}\bigl(\Phi_\mu(z)\bigr)
 =\mu\frac{4BD}{aL(z)}+O_K(\mu^2).
\]
The independence of this first variation from the chosen analytic resonance
identification was established in
Remark~\ref{rem:intrinsic-first-variation}.  At
\((a,B,D)=(-1/8,1,1)\), multiplication by \(1/(4\mu)\) gives
\(-8/L(z)+O_K(\mu)\).  Combining this with the strong matching proves
\eqref{eq:Euler-Dysthe-normalized-match}.
\end{proof}

As an exact normalization check, the selected family
\(p=(-5,-2,7)\), \(q=(-7,2,5)\) has \(L=14\), so the theorem gives
\[
 \frac\mu4W_6^{\rm E}
 =-\frac47+O(\mu),\qquad
 \frac1{4\mu}\widehat W_6^{\rm D}
 =-\frac47+O(\mu).
\]
The numerical consistency checks in Online Resource~1 reproduce this common
leading coefficient on the displayed family and on a second asymmetric family.
The uniform asymptotic statement follows from the analytic proof above.

\section{Singular limit of the cancellation loci}
\label{sec:bifurcation}

Assume \(h\ne0\).  Fix \(\eps\ne0\) and a spectral bound \(K>0\).  For
\(\lambda\ne0\), the genuine four-wave plane is
\begin{equation}
 S=-\frac{2a}{3\lambda\eps h}.
 \label{eq:lambda-plane}
\end{equation}
Every quartet in \([-K,K]^4\) has \(\abs S\le2K\).  Hence, if
\begin{equation}
 \abs\lambda<\frac{\abs a}{3K\abs{\eps h}},
 \label{eq:escape-bound}
\end{equation}
the genuine plane does not meet that compact spectral box.  On every fixed
bounded wavenumber set, the temporal quartic resonance variety therefore
reduces to the spatial action-only component for sufficiently small
\(\abs\lambda\).  Globally, the plane exists for every \(\lambda\ne0\).  This
is the precise meaning of ``the resonant plane escapes to infinity.''

In coefficient space \((B,C,D,E)\), the spatial cancellation set is
\begin{equation}
 \mathcal C_{\mathrm{sp}}
 =\{D=E=0\}
 \cup
 \left\{B=C=0,\ E=-\frac{D^2}{4a}\right\}.
 \label{eq:Csp}
\end{equation}
For \(\lambda\ne0\), simultaneous temporal four- and six-wave cancellation is
\begin{equation}
 \mathcal C_{\mathrm{temp}}(\lambda)
 =\left\{D=E=0,\ C=\frac{3\lambda h}{a}B\right\}.
 \label{eq:Ctemp}
\end{equation}
For every \(R>0\), the intersections
\(\mathcal C_{\mathrm{temp}}(\lambda)\cap\overline B_R(0)\)
converge as \(\lambda\to0\) to
\(\mathcal C_{\mathrm{lim}}\cap\overline B_R(0)\), in the Hausdorff distance
on compact sets, where
\begin{equation}
 \mathcal C_{\mathrm{lim}}=\{C=D=E=0\},
 \label{eq:Clim}
\end{equation}
a proper subset of the spatial local component.  The spatial nonlocal
Calogero--Moser component with \(D\ne0\) has no temporal continuation because
the temporal four-wave condition already forces \(D=0\).

Thus the additional resonance plane leaves every bounded spectral window,
while the limiting temporal cancellation set is a proper subset of the
spatial one.

\section{Computer-assisted verification}
\label{sec:verification}

The temporal necessity calculation uses the exact determinant and ideal
certificate in Appendix~\ref{app:temporal-certificate}.  The spatial chamber
identities and the diagonal Laurent expansion were checked in exact rational
arithmetic with SymPy \cite{MeurerEtAl2017}.  Online Resource~1 contains the
scripts, exact outputs, requirements, and reproduction instructions.  Its
Euler calculations are numerical consistency checks; the uniform Euler
matching and bounded-remainder estimates are proved in Section~\ref{sec:euler-matching}.

\section{Discussion and conclusion}
\label{sec:discussion}

The leading Dysthe coefficient \(4BD/(aL)\) couples the local cubic
interaction to the induced mean flow.  Its persistence on compact resonance
patches shows that the obstruction is not tied to a selected sextet.
At the Craig--Guyenne--Sulem coefficients, the complete Euler coefficient
has the same normalized limit \(-8/L\).  The nine near-resonant quartic
contractions produce the singular term; every other connected degree-six
contribution remains bounded.

The Fedele--Dutykh spatial models and the Craig--Guyenne--Sulem temporal
model lie outside their cancellation sets.  Under the assumptions of
Section~\ref{sec:ZS}, every \(C^3\) quadratic leading symbol is therefore in
\(\operatorname{span}\{1,k,\omega(k)\}\).  These models cannot support an
inverse-scattering hierarchy with infinitely many linearly independent
regular quadratic leading symbols.

On the cancellation sets, the coefficient relations contain the known mixed
Chen--Lee--Liu, Calogero--Moser derivative nonlinear Schr\"odinger, and
Hirota representatives described in
Corollary~\ref{cor:integrable-branches}.  Their integrability is supplied by
the cited inverse-scattering and Lax-pair results.

The limit of small cubic dispersion depends on the spectral scale.  The
six-wave coefficient varies uniformly on fixed compact resonance patches,
whereas the additional four-wave plane moves to unbounded wavenumber.
Accordingly, the limiting temporal cancellation set is smaller than the
spatial cancellation set computed directly at zero cubic dispersion.

\clearpage
\appendix
\section{Derivation of the diagonal collision expansion}
\label{app:diagonal-expansion}

This appendix derives \eqref{eq:diagonal-W} directly from the nine channel
formula \eqref{eq:temp-kernel}.  Put
\[
 P=(-5,-2,7),\qquad Q=(-7,x,7-x),
\]
so that \(p_\ell=\kappa+\rho P_\ell\) and
\(q_\alpha=\kappa+\rho Q_\alpha\), with \(\rho>0\).  Retain
\[
 A=A_\kappa=a+3\gamma\kappa,\qquad
 U=U_\kappa=B+c\kappa.
\]
The exact resonance equation is
\[
 -2A(x-2)(x-5)+21\gamma\rho(10+7x-x^2)=0.
\]
Substitution of
\(x=2+x_1\rho+x_2\rho^2+O(\rho^3)\) and comparison of the first two
orders gives
\begin{equation}
 x_1=-\frac{70\gamma}{A},\qquad
 x_2=\frac{7105\gamma^2}{3A^2}.
 \label{eq:diagonal-x-second-order}
\end{equation}

At \(\rho=0\), the ordered offsets are
\[
 P=(-5,-2,7),\qquad Q^0=(-7,2,5),
\]
and the signs of \(Q_\alpha-P_m\) are
\begin{equation}
 \sigma_{\alpha m}
 =\operatorname{sgn}(Q_\alpha^0-P_m)
 =
 \begin{pmatrix}
 -1&-1&-1\\
  1& 1&-1\\
  1& 1&-1
 \end{pmatrix}.
 \label{eq:diagonal-sign-matrix}
\end{equation}
These signs remain fixed for sufficiently small positive \(\rho\).  If
\(\{i,j\}=\{1,2,3\}\setminus\{\ell\}\), the resolved vertices and divisor are
\begin{align}
 V^{(1)}_{\alpha\ell}
 &=
 U-\frac{\rho}{2}\left[
 cP_\ell+d\sum_{m\ne\ell}
 \sigma_{\alpha m}(Q_\alpha-P_m)\right],
 \label{eq:diagonal-V1-expanded}\\
 V^{(2)}_{\alpha\ell}
 &=
 U-\frac{\rho}{2}\left[
 cQ_\alpha+d\sum_{n\ne\alpha}
 \sigma_{n\ell}(Q_n-P_\ell)\right],
 \label{eq:diagonal-V2-expanded}\\
 \Delta_{\ell\alpha}
 &=
 \rho^2(Q_\alpha-P_i)(Q_\alpha-P_j)
 (2A-3\gamma\rho P_\ell).
 \label{eq:diagonal-divisor-expanded}
\end{align}
Thus every term required for the Laurent expansion is explicit and rational
in \(A,\gamma,U,c,d\).

Let
\[
 \mathcal A_{\rm temp}
 =\sum_{\ell,\alpha}
 \frac{V^{(1)}_{\alpha\ell}V^{(2)}_{\alpha\ell}}
 {\Delta_{\ell\alpha}}.
\]
Substitution of \eqref{eq:diagonal-x-second-order}--%
\eqref{eq:diagonal-divisor-expanded} and summation over all nine channels
gives the coefficients in Table~\ref{tab:diagonal-Laurent}.  The table also
exhibits the exact cancellation of the apparent \(\rho^{-2}\) term.
\begin{table}[htbp]
\centering
\begin{tabular}{@{}cc@{}}
\toprule
Power & Coefficient in \(\mathcal A_{\rm temp}\)\\
\midrule
\(\rho^{-2}\) & \(0\)\\[0.2em]
\(\rho^{-1}\) & \(\displaystyle-\frac{Ud}{14A}\)\\[0.7em]
\(\rho^0\) &
\(\displaystyle
 \frac{27\gamma^2U^2}{8A^3}
 -\frac{9\gamma Uc}{8A^2}
 +\frac{3d^2}{4A}\)\\
\bottomrule
\end{tabular}
\caption{Exact Laurent coefficients of the summed nine-channel contraction
on the diagonal collision family.}
\label{tab:diagonal-Laurent}
\end{table}

Finally, \(W_6=-4\mathcal A_{\rm temp}-12e\), and
\[
 cA-3\gamma U
 =c(a+3\gamma\kappa)-3\gamma(B+c\kappa)
 =ac-3\gamma B.
\]
Therefore
\[
 W_6(p_\rho;q_\rho)
 =
 \frac{2dU_\kappa}{7A_\kappa\rho}
 +\frac{9\gamma U_\kappa(ac-3\gamma B)}{2A_\kappa^3}
 -\frac{3d^2}{A_\kappa}-12e+O(\rho),
\]
which is \eqref{eq:diagonal-W}.  In the verification directory, the exact
script \path{verify_diagonal_channel_expansion.py} reconstructs
\eqref{eq:diagonal-V1-expanded}--\eqref{eq:diagonal-divisor-expanded} and
checks each coefficient in Table~\ref{tab:diagonal-Laurent} over the
characteristic-zero rational-function field.

\section{Exact temporal necessity certificate}
\label{app:temporal-certificate}

This appendix makes the finite computer-assisted part of
Theorem~\ref{thm:temporal} directly checkable.  After shifting the dispersion
to pure cubic form, write
\(B_0=B-ac/(3\gamma)\), \(C_0=c\), \(D_0=d\), and \(E_0=\gamma e\).
The shifted dispersion is \(\gamma y^3\) plus affine terms.  Thus the
normalized equation is
\[
 \gamma W_6=-4\sum_{\ell,\alpha}
 \frac{T^{(1)}_{\ell\alpha}T^{(2)}_{\ell\alpha}}
 {\Delta_{\ell\alpha}/\gamma}-12E_0=0.
\]
The six exact resonances are
\begin{equation}
\begin{gathered}
(-6,-3,5)\to(-5,-1,2),\qquad
(-6,-2,5)\to(-5,-1,3),\\
(-6,1,5)\to(-5,2,3),\qquad
(-5,-1,6)\to(-3,-2,5),\\
(-5,2,6)\to(-3,1,5),\qquad
(-5,3,6)\to(-2,1,5).
\end{gathered}
\label{eq:temporal-witnesses}
\end{equation}
Four witnesses have repeated slots within a virtual quartet, although
their six external modes are distinct and every virtual divisor is nonzero.
We evaluate the continuous extension of the generic kernel at those
points.  To justify this, fix the incoming triple and vary the outgoing
triple on its pure-cubic resonance curve.  A tangent is
\[
 (q_2^2-q_3^2,\ q_3^2-q_1^2,\ q_1^2-q_2^2).
\]
For each listed witness its components are nonzero.  The repeated-slot
equalities \(r_{\ell\alpha}=q_\alpha\) or
\(r_{\ell\alpha}=p_\ell\) therefore break under a small variation, while
external separation and the nonzero divisors persist.  The witnesses
are thus limits of generic resonances.  Identical vanishing on the
generic set implies vanishing at these limits.
In channel order
\((\ell,\alpha)=(1,1),(1,2),\ldots,(3,3)\), the corresponding nine
normalized denominators \(\Delta_{\ell\alpha}/\gamma\) are the rows
\begin{equation}
\left(\begin{smallmatrix}
120&-72&-90&30&90&72&54&-270&-1080\\
270&-54&-90&30&90&54&72&-120&-1080\\
1080&-54&-72&30&72&54&90&-120&-270\\
270&120&-90&-54&-72&-30&72&54&-1080\\
1080&120&-72&-54&-90&-30&90&54&-270\\
1080&270&-54&-72&-90&-30&90&72&-120
\end{smallmatrix}\right).
\label{eq:temporal-denominator-table}
\end{equation}
Every entry is nonzero, and the only dispersion factor divided out is the
assumed nonzero \(\gamma\).  Clearing the displayed integer denominators
gives the following primitive polynomials, with
\(\gamma W_6=f_j/270\) at witness \(j\):
\begin{align}
f_1={}&-60B_0^2+80B_0C_0+511B_0D_0-327C_0D_0
       -1018D_0^2-3240E_0,\notag\\
f_2={}&-45B_0^2+45B_0C_0+370B_0D_0-150C_0D_0
       -686D_0^2-3240E_0,\notag\\
f_3={}&-36B_0^2+283B_0D_0+87C_0D_0-464D_0^2-3240E_0,\notag\\
f_4={}&36B_0^2-283B_0D_0+87C_0D_0+464D_0^2-3240E_0,\notag\\
f_5={}&45B_0^2+45B_0C_0-370B_0D_0-150C_0D_0
       +686D_0^2-3240E_0,\notag\\
f_6={}&60B_0^2+80B_0C_0-511B_0D_0-327C_0D_0
       +1018D_0^2-3240E_0.
\label{eq:temporal-cleared-polynomials}
\end{align}
The simplest certificate is linear algebra.  In the ordered monomial list
\[
(B_0^2,B_0C_0,B_0D_0,C_0D_0,D_0^2,E_0),
\]
the coefficient matrix of \((f_1,\ldots,f_6)\) has determinant
\(59584377600\ne0\).  Its rational inverse expresses each monomial as a
linear combination of the six witness polynomials.  Consequently
\begin{equation}
\boxed{
 I=\langle f_1,\ldots,f_6\rangle
 =\langle B_0^2,B_0C_0,B_0D_0,C_0D_0,D_0^2,E_0\rangle.
}
\label{eq:temporal-ideal-certificate}
\end{equation}
With variable order \((E_0,B_0,C_0,D_0)\) and graded reverse
lexicographic order, the displayed monomial generators are also the exact
reduced Gr\"obner basis:
\begin{equation}
\{B_0^2,B_0C_0,B_0D_0,C_0D_0,D_0^2,E_0\}.
\label{eq:temporal-groebner-certificate}
\end{equation}
The zero locus of this ideal, over either \(\mathbb R\) or \(\mathbb C\), is
\[
 V(I)=\{B_0=D_0=E_0=0\},\qquad C_0\ \text{arbitrary}.
\]
Indeed, the square generators force \(B_0=D_0=0\), and the separate
generator \(E_0\) forces \(E_0=0\).  Conversely, every listed generator
vanishes on this set.  Equivalently, \(\sqrt I=\langle B_0,D_0,E_0\rangle\).
The scripts in Online Resource~1 reconstruct the witnesses, denominators,
vertices, cleared polynomials, and linear-algebra certificate in exact
arithmetic.

Online Resource~1 also contains the Euler coefficient calculation and its
numerical consistency checks.  The quartic-kernel and
bounded-remainder arguments are given in
Lemmas~\ref{lem:Euler-quartic-match} and \ref{lem:Euler-bounded-remainder}.

\section*{Acknowledgements}

The authors thank Thomas A. Schmidt for reviewing the Gr\"obner-basis
calculation, and Thomas A. Schmidt and Patrik V. Nabelek for discussions
and presentations in our joint nonlinear waves seminar.
Alex J. Sutherland gratefully acknowledges Oregon State University for
support through a graduate teaching assistantship, which made this work
possible.

\section*{Statements and Declarations}

\noindent\textbf{Funding.}
Alex J. Sutherland was supported by a graduate teaching assistantship at
Oregon State University.

\medskip
\noindent\textbf{Competing interests.}
The authors declare no competing interests.

\medskip
\noindent\textbf{Use of AI-assisted technologies.}
During preparation of this work, the authors used OpenAI Codex (accessed
August--September 2026) to assist in identifying candidate rational resonance
witnesses and in developing and checking portions of the symbolic verification
code, including the Gr\"obner-basis and Euler-coefficient calculations.  The
exact symbolic certificates and separate numerical regressions are provided
with the article.  The authors take full responsibility for the results and
their interpretation.

\medskip
\noindent\textbf{Data and code availability.}
All calculations supporting the results are contained in the article and the
\emph{Supplementary Verification Archive} (Online Resource~1), which provides
the exact symbolic scripts, outputs, requirements, reproduction instructions,
and a SHA-256 manifest.  The same archive is included with the arXiv source.

\medskip
\noindent\textbf{Online Resource 1.}
Symbolic verification scripts and exact certificates for the Dysthe
classification and resonance-deformation calculations, with the Euler
coefficient calculation and separate numerical consistency checks.

\begingroup
\small
\bibliographystyle{plainnat}
\bibliography{references}
\endgroup

\end{document}